\documentclass[11pt]{article}
\usepackage[utf8]{inputenc}
\usepackage[left=1in, top=1in, right=1in, bottom=1in]{geometry}
\usepackage{natbib}
\usepackage{hyperref}
\usepackage{mathtools}
\usepackage{amsmath, amsthm, amsfonts}
\usepackage{bm, bbm}
\usepackage{subcaption}

\newcommand{\eps}{\varepsilon}

\newtheorem{theorem}{Theorem}[section]

\newtheorem{lemma}[theorem]{Lemma}
\newtheorem{definition}{Definition}[section]
\newtheorem{claim}[theorem]{Claim}

\newtheorem{example}{Example}[section]
\newtheorem{assumption}{Assumption}[section]

\DeclareMathOperator*{\argmax}{arg\,max}
\DeclareRobustCommand\iff{\;\Longleftrightarrow\;}

\DeclareMathOperator{\OPT}{OPT}

\newcommand{\E}{\mathbb E}
\newcommand{\smid}{\,|\,}

\newcommand{\dTV}{d_{\mathrm{TV}}}

\newcommand{\prior}{\mu}
\newcommand{\su}{u}
\newcommand{\esu}{U}
\newcommand{\ru}{v}
\newcommand{\BP}{\mathrm{BP}}
\newcommand{\OBJ}{\mathrm{OBJ}}
\newcommand{\adv}{\mathrm{adv}}
\newcommand{\ob}{\mathrm{ob}}

\newcount\Comments  
\usepackage{color} 
\newcommand{\kibitz}[2]{\ifnum\Comments=1\textcolor{#1}{#2}\fi \ignorespaces}
\newcommand{\yc}[1]{\kibitz{cyan} {[YC: #1]}}
\newcommand{\tao} [1]  {\kibitz{blue}{\bf\noindent [Tao: #1]} }

\title{Persuading a Behavioral Agent: \\
Approximately Best Responding and Learning}

\author{
Yiling Chen\\
{\small Harvard University} \\
{\small \tt{yiling@seas.harvard.edu}}
\and
Tao Lin\\
{\small Harvard University} \\
{\small \tt{tlin@g.harvard.edu}}
}
\date{\small
First version: Feb, 2023. \\
{\color{red} The main results of this draft have been subsumed by our later paper \cite{lin2024persuading}: \href{https://arxiv.org/pdf/2402.09721.pdf}{link}} }

\begin{document}

\maketitle

\begin{abstract}
The classic Bayesian persuasion model assumes a Bayesian and best-responding receiver.
We study a relaxation of the Bayesian persuasion model where the receiver can approximately best respond to the sender's signaling scheme.
We show that, under natural assumptions, (1) the sender can find a signaling scheme that guarantees itself an expected utility almost as good as its optimal utility in the classic model, no matter what approximately best-responding strategy the receiver uses; (2) on the other hand, there is no signaling scheme that gives the sender much more utility than its optimal utility in the classic model, even if the receiver uses the approximately best-responding strategy that is best for the sender.
Together, (1) and (2) imply that the approximately best-responding behavior of the receiver does not affect the sender's maximal achievable utility a lot in the Bayesian persuasion problem. 
The proofs of both results rely on the idea of robustification of a Bayesian persuasion scheme: given a pair of the sender's signaling scheme and the receiver's strategy, we can construct another signaling scheme such that the receiver prefers to use that strategy in the new scheme more than in the original scheme, and the two schemes give the sender similar utilities.
As an application of our main result (1), we show that, in a repeated Bayesian persuasion model where the receiver learns to respond to the sender by some algorithms, the sender can do almost as well as in the classic model.  Interestingly, unlike (2), with a learning receiver the sender can sometimes do much better than in the classic model. 
\end{abstract}


\section{Introduction}
Bayesian persuasion, proposed by \citet{kamenica_bayesian_2011} and extensively studied in the literature in recent years (see, e.g., surveys by \citet{dughmi_algorithmic_2017}, \citet{kamenica2019bayesian}, and \citet{bergemann_information_2019}),  is an information design problem that studies how a person with an informational advantage (sender) should reveal information to influence the decision-making of another person (receiver).  In the classic model, the sender knows a state of the world that is unknown to the receiver, sends a signal to the receiver, and the receiver responds by taking some action that affects the utilities of both parties.   
An important assumption in the classic model is that the receiver always best responds: namely, given the signal, the receiver takes an action that is optimal for itself with respect to the posterior distribution of the state of the world. 
This assumption requires, for example, the receiver to have an accurate prior belief about the state of the world, to know the sender's signaling scheme, to follow Bayes' rule, and to know its own utility function exactly. Not all of these requirements are likely to be satisfied in practice.  For example, evidence from behavioral economics shows that human decision-makers often fail to be perfectly Bayesian \citep{camerer1998bounded, benjamin2019errors}.  

In this work, we consider a Bayesian persuasion model where the receiver is allowed to \emph{approximately} best respond.
Specifically, upon receiving a signal from the sender, the receiver can take an action that is approximately optimal for itself with respect to the posterior distribution of the state of the world.   If there are multiple approximately optimal actions, the receiver can take an arbitrary one. 
Under this model, we characterize the possible range of expected utility the sender can achieve by some signaling scheme when the receiver approximately best responds in different ways.  In particular, we consider two ways: the worst and the best (for the sender). 

\paragraph{\bf The worst way of approximately best responding}
The first question we ask is: if the receiver approximately best responds in the \emph{worst} way for the sender, what is the maximal utility the sender can achieve by some signaling scheme?
This is a ``maximin'' optimization problem and belongs to the literature of ``robust Bayesian persuasion'' (e.g., \cite{dworczak2020preparing, de_clippel_non-bayesian_2022}).
\yc{We haven't talked about what is a robust Bayesian persuasion problem.}
\tao{Shall we just delete ``robust Bayesian persuasion'' or reference to the related works?}\yc{Maybe reference it and point it to the related work}
The following example shows that, in general, the sender's optimal utility against a worst-case approximately best-responding receiver can be really bad compared to the sender's optimal utility in the classic Bayesian persuasion model:

\begin{example}\label{ex:no-robust-scheme}
There are two possible states of the world, $\omega_1, \omega_2$.  The receiver can take two actions $a_1, a_2$.  The utilities of the sender and the receiver when the receiver takes an action at a state are shown below: 
\begin{center}
\begin{tabular}{|l|l|l|lll|l|l|l|}
\cline{1-3} \cline{7-9}
sender & $\omega_1$  & $\omega_2$  &  & \hspace{1em} &  & receiver & $\omega_1$ & $\omega_2$ \\ \cline{1-3} \cline{7-9} 
$a_1$ & $1$         & $0$         &  &  &  & $a_1$ & $0$        & $0$        \\ \cline{1-3} \cline{7-9} 
$a_2$ & $0$ & $0$ &  &  &  & $a_2$ & $0$        & $1$        \\ \cline{1-3} \cline{7-9}
\end{tabular}
\end{center}
In words, the sender obtains positive utility only when the receiver takes action $a_1$ at state $\omega_1$ and the receiver obtains positive utility only when taking action $a_2$ at state $\omega_2$. 
The sender and the receiver share a common prior belief of $\,\Pr[\omega_1] = \Pr[\omega_2] = 1/2$. 

In the classic Bayesian persuasion model \citep{kamenica_bayesian_2011}, an optimal signaling scheme for the sender is to fully reveal the state.  The receiver responds by taking action $a_1$ at state $\omega_1$ and taking action $a_2$ at state $\omega_2$.\footnote{At state $\omega_1$ the receiver is indifferent between the two actions, and the classic model assumes that the receiver breaks ties in favor of the sender.}
The sender's expected utility is $1/2$. 

In our model where the receiver approximately best responds in the worst way for the sender, no matter what signaling scheme is used and what signal is sent, $a_2$ is always a weakly better action than $a_1$ for the receiver, so the receiver can always take $a_2$. The sender's utility becomes $0$, significantly less than the $1/2$ in the classic model. 
\end{example}

Our first main result finds a natural condition under which the phenomenon in Example \ref{ex:no-robust-scheme} will not happen and \emph{there exists a signaling scheme that guarantees the sender a utility that is almost as good as the sender's optimal utility in the classic Bayesian persuasion model, even if the receiver approximately best responds in the worst way}.  Roughly speaking, this condition is: for the receiver, for any possible action $a$, there exists a state of the world $\omega$ for which $a$ is the unique optimal action.  This condition is not satisfied by Example \ref{ex:no-robust-scheme} and is satisfied by the classic ``judge example'' in \cite{kamenica_bayesian_2011}.  We use the judge example to illustrate the main idea:

\begin{example}
There are a prosecutor (sender), a judge (receiver), and a defendant.  There are two possible states of the world: the defendant is guilty or innocent.  The judge has to decide to either convict or acquit the defendant.  As shown by the utility matrices below, the prosecutor always wants the judge to convict while the judge wants to convict if and only if the defendant is guilty. 
\begin{center}
\begin{tabular}{|l|l|l|lll|l|l|l|}
\cline{1-3} \cline{7-9}
prosecutor  & guilty  & innocent  &  & \hspace{1em} &  & judge & guilty & innocent \\ \cline{1-3} \cline{7-9} 
convict & $1$         & $1$         &  &  &  & convict & $1$        & $0$        \\ \cline{1-3} \cline{7-9} 
acquit & $0$ & $0$ &  &  &  & acquit & $0$        & $1$        \\ \cline{1-3} \cline{7-9}
\end{tabular}
\end{center}
The judge's utility satisfies our condition: convict is uniquely optimal for guilty and acquit is uniquely optimal for innocent. 

The prosecutor and the judge share a prior belief of $\mu = \Pr[\text{guilty}] = 0.3$.

In the classic model, the optimal signaling scheme for the prosecutor is the following: if the defendant is guilty, say ``$g$'' with probability $1$; if the defendant is innocent, say ``$i$'' with probability $\frac{4}{7}$ and ``$g$'' with probability $\frac{3}{7}$.  Under this signaling scheme, when hearing ``$i$'' the judge knows for sure that the defendant is innocent and will acquit.  When hearing ``$g$'' the judge only knows that the defendant is guilty with probability $\frac{0.3 \times 1}{0.3 \times 1 + 0.7 \times \frac{3}{7}} = 0.5$, so the judge is indifferent between conviction and acquitting and is assumed to break ties in favor of the prosecutor (convicts).
So, the expected utility of the prosecutor under this signaling scheme is $\Pr[i] \times 0 + \Pr[g] \times 1 = 0.3 \times 1 + 0.7 \times \frac{3}{7} = 0.6$.

The above signaling scheme is not a good scheme for our model where the receiver can approximately best respond in the worst way for the sender, because the judge can choose to acquit when hearing ``$g$'', giving the prosecutor a utility of $0$.

A simple way to find a good signaling scheme for our model is to perturb the optimal scheme in the classic model.  In the above example, by slightly decreasing the probability with which the prosecutor says ``$g$'' for an innocent defendant (the $\frac{3}{7}$ probability), the judge will have a $0.5+\eps$ belief for the defendant being guilty when hearing ``$g$''.  If $\eps$ is large enough such that there is a substantial gap between the utilities of conviction and acquitting for the judge, then the judge will continue to choose to convict even if the judge best responds only approximately.  This perturbed scheme gives the prosecutor an expected utility of $0.6 - O(\eps)$.
\end{example}

The idea of finding a good scheme for our approximately-best-responding model by perturbing the optimal scheme in the classic model is general.  It can be applied to all Bayesian persuasion instances that satisfy the condition we identified above, hence proving the result that a scheme that is almost as good as the optimal scheme in the classic model exists.   We call this perturbation idea \emph{robustification of a Bayesian persuasion scheme}, meaning that the optimal yet sensitive scheme in the classic model can be turned into a robust scheme against the receiver's worst-case behavior, without changing the sender's expected utility a lot.

\paragraph{\bf The best way of approximately best responding}
The second question we ask is: if the receiver approximately best responds in the \emph{best} way for the sender, what is the maximal utility the sender can achieve by some signaling scheme?
This is a ``maximax'' objective that quantifies the maximal extent to which the sender can exploit the receiver's irrational behavior.
One may initially think that the sender now can do much better than the classic Bayesian persuasion model.  However, we show that, again under the condition we identified above, \emph{the sender cannot do much better than the classic model}.  Namely, there is no signaling scheme that gives the sender a much higher utility than the sender's optimal utility in the classic model, even if the receiver approximately best responds in favor of the sender.
Together with the first main result, we conclude that, in the Bayesian persuasion problem, the sender's maximal achievable utility is not affected a lot by the receiver's approximately best-responding behavior. 
Interestingly, our second main result is also proved by the perturbation/robustification idea used in the proof of the first result.  

\paragraph{\bf Application: persuading a learning receiver}
Finally, we apply our techniques developed above (in particular, the robustification idea) to a model of repeated Bayesian persuasion where the receiver uses some algorithms to learn to respond to the sender.  The learning problem of the receiver is a contextual multi-armed bandit problem where algorithms like Multiplicative Weight Update and EXP-3 can be applied.  At a high level, the learning behavior of the receiver is also an approximately-best-responding behavior, hence our techniques for the approximately-best-responding model can be applied here to prove a result similar to our first main result above: \emph{with a learning receiver, the sender can always do almost as well as in the classic non-learning model}.
Interestingly, our second main result above no longer holds in the learning receiver model: we find an example with a learning receiver where the sender can do \emph{much better} than in the classic non-learning model.
Finding conditions under which the sender cannot do much better than in the classic model is an interesting future direction.

\subsection{Related Works}\label{sec:related}
Since the seminal work of \citet{kamenica_bayesian_2011} (who generalized the model of \citet{brocas_influence_2007}), the Bayesian persuasion problem has inspired a long and active line of research in the literature.  
Extensions like Bayesian persuasion with multiple senders \citep{gentzkow_bayesian_2017, ravindran_competing_2022} and multiple receivers \citep{wang_bayesian_2013, taneva_information_2019} are studied.
Applications to auctions \citep{emek_signaling_2014} and voting \citep{alonso_persuading_2016} are considered.  
The computational complexity of Bayesian persuasion is determined \citep{dughmi_algorithmic_2016}.
Interested readers can refer to the surveys by \citet{dughmi_algorithmic_2017}, \citet{kamenica2019bayesian}, and \citet{bergemann_information_2019} for excellent overviews of the literature.

Closest to our work is the line of works on \emph{Robust Bayesian persuasion}, which study the Bayesian persuasion problem with various assumptions in the classic model being relaxed.
For the sender, the sender does not know the receiver's prior belief \citep{kosterina_persuasion_2022, zu_learning_2021}, does not know the utility function of the receiver \citep{babichenko_regret-minimizing_2021}.
For the receiver, the receiver may receive additional signals besides the signal sent by the sender \citep{dworczak2020preparing, ziegler2020adversarial}, make mistakes in Bayesian update \citep{de_clippel_non-bayesian_2022}, 
be risk-conscious \citep{anunrojwong_persuading_2023}, 
do quantal response \citep{feng2024rationality} or approximate best response \citep{yang2024computational}.
Many of these works take the maximin approach of maximizing the sender's worst-case utility when the classic assumptions are violated, aiming to find the exact solution to the maximin problem.  
We also take a maximin approach, but unlike previous works, instead of finding the exact maximin solution, we are interested in bounding the maximin solution and finding conditions under which the maximin solution is close to the original solution.  This allows us to study a very general model that captures several previous models as special cases (e.g., \cite{de_clippel_non-bayesian_2022} and \cite{feng2024rationality}).  We also study a maximax problem, which is typically not studied in the robust Bayesian persuasion literature.  

Our work also relates to a recent trend of works on Bayesian persuasion with learning: \cite{castiglioni_online_2020, castiglioni2021multi, zu_learning_2021, feng_online_2022, wu_sequential_2022}.
The learner in these works is the sender, who repeatedly interacts with the receiver to learn missing information and adjust the signaling scheme over time.  In our work, however, the learner is the receiver, who learns to respond to the sender.
In this respect, our work has a similar spirit with \cite{braverman_selling_2018, camara_mechanisms_2020, cai2023selling, rubinstein2024strategizing, deng_strategizing_2019, mansour_strategizing_2022, guruganesh2024contracting} which study learning agents in auction design,  Stackelberg games, and contract design.  They find that the principal (seller, leader) can exploit the learning behavior of the agent (buyer, follower) to obtain a higher utility than in the non-learning model.  We find a similar phenomenon in the Bayesian persuasion setting.

Our work situates in a broad agenda of design for behavioral agents. As the Bayesian rational agent model has been shown to be violated in experiments and in practice \citep{10.1257/aer.103.3.617}, many behavioral models of agents have been proposed (e.g., \cite{Stahl1995OnPM, 10.1162/0033553041502225}). While they each explain away some of the observed rationality violations for specific settings, they are similarly observed to not hold perfectly in the real world. Thus, one challenge of mechanism and information design for behavioral agents is that there is no perfect model. Our approximately-best-responding agent model includes more than one behavioral model (e.g., the quantal response model \citep{mckelvey_quantal_1995}) and, with it, our approach of bounding the sender's expected utility in persuasion provides a more robust guarantee than finding the optimal persuasion strategy with respect to a specific behavioral model.

\tao{Hi Yiling.  I might need your help on this part.  What are some references for this part?}\yc{I added the previous paragraph, but maybe it fits better in the intro before the related works.}

\section{Model: Approximately Best Responding Agent}
\label{sec:model}
Let $\Omega$ be a finite set of \emph{states of the world} (or simply \emph{state}), with $|\Omega| = m$. 
A sender (persuader) and a receiver (decision maker) have a common prior belief $\prior \in \Delta(\Omega)$ about the state of the world $\omega \in \Omega$. We use $\Delta(X)$ to denote the set of all probability distributions over set $X$. Let $S$ be a finite set of \emph{signals}.
The sender first commits to a \emph{signaling scheme} $\pi: \Omega \to \Delta(S)$, which maps each possible state $\omega\in \Omega$ to a distribution over signals.
Then, a state $\omega$ is realized according to $\prior$. 
The sender observes $\omega$ and sends a signal $s \sim \pi(\omega)$ to the receiver.
We use $\pi(s|\omega)$ to denote the conditional probability that the sender sends signal $s$ at state $\omega$, use $\pi(\omega, s) = \prior(\omega) \pi(s|\omega)$ to denote the joint probability that the state is $\omega$ and the signal is $s$, and use $\pi(s) = \sum_{\omega\in\Omega} \pi(\omega, s)$ to denote the marginal probability of $s$.
The receiver has a finite set of \emph{actions} $A$.  When the receiver takes action $a\in A$ and the state is $\omega\in\Omega$, the receiver obtains utility $\ru(a, \omega)$ and the sender obtains utility $\su(a, \omega)$. 
We normalize the utility so that $0\le \ru(a, \omega), \su(a, \omega) \le 1$ for all $a\in A$, $\omega \in \Omega$. 

\paragraph{\bf Receiver's behavior}
The receiver knows the sender's signaling scheme $\pi$ but not the realization of $\omega$. Upon receiving some signal $s\in S$, the receiver forms a posterior belief $\mu_s \in \Delta(\Omega)$ about $\omega$ according to Bayes' rule $\mu_s(\omega) = \pi(\omega|s) = \frac{\prior(\omega)\pi(s|\omega)}{\pi(s)}$.
Overloading notations, we let $\ru(a, \mu_s)$ denote the receiver's expected utility when taking action $a$ with belief $\mu_s$: $\ru(a, \mu_s) = \E_{\omega \sim \mu_s}[ \ru(a, \omega) ] = \sum_{\omega \in \Omega} \pi(\omega | s) \ru(a, \omega)$. 

Classic works on Bayesian persuasion usually assume a \emph{best-responding} receiver: after receiving a signal $s$ (and forming the posterior $\mu_s$), the receiver picks an action $a^*_\pi(s) \in \argmax_{a\in A} \ru(a, \mu_s)$ to maximize its expected utility.
In this work, we allow the receiver to \emph{approximately} best respond.
Let $\gamma \ge 0$ be a parameter.  We say an action $a\in A$ is a \emph{$\gamma$-best-responding action} to signal $s$ (or posterior $\mu_s$) if the expected utility of action $a$ given signal $s$ is at most $\gamma$-worse than a best-responding action: 
\begin{equation}
    \ru(a, \mu_s) \ge \ru(a^*_\pi(s), \mu_s) - \gamma.
\end{equation}
Let $A_\pi^\gamma(s)$ be the set of all $\gamma$-best-responding actions to $s$ (under signaling scheme $\pi$): 
\begin{equation} 
    A_\pi^\gamma(s) \coloneqq \Big\{ a\in A ~\big|~ \ru(a, \mu_s) \ge \ru(a^*_\pi(s), \mu_s) - \gamma \Big\}. 
\end{equation}
Clearly, $A_\pi^0(s)$ is the set of best-responding actions and $a^*_\pi(s) \in A_\pi^0(s) \subseteq A_\pi^\gamma(s)$.
A (randomized) \emph{receiver strategy} is denoted by $\rho: S \to \Delta(A)$, a mapping from every signal $s$ to a distribution over actions $\rho(s)$.  We overload the notation $\rho(s)$ to also denote the randomized action $a$ distributed according to the distribution $\rho(s)$.  An approximately best-responding receiver strategy is a strategy where the receiver takes approximately best-responding actions with high probability, formalized as follows:   
\begin{definition}
Let $0\le \gamma, \delta < 1$ be two parameters. 
We say a receiver strategy $\rho$ is a 
\begin{itemize}
\item \emph{best-responding strategy}, if $~\Pr[\rho(s) \in A_\pi^0(s)] = 1$ for every $s\in S$ with $\pi(s) > 0$;
\item \emph{$\gamma$-best-responding strategy}, if $~\Pr[\rho(s) \in A_\pi^\gamma(s)] = 1$ for every $s\in S$ with $\pi(s) > 0$; \yc{Say $P (\rho(s) \in A_\pi^0(s)) = 1$ instead? Similar for the two items below. } \tao{Ok!}
\item \emph{$(\gamma, \delta)$-best-responding strategy}, if $~\Pr[\rho(s)\in A_\pi^\gamma(s)] \ge 1-\delta$ for every $s\in S$ with $\pi(s) > 0$.  
\end{itemize}
\end{definition}
The set of best-responding strategies is a subset of $\gamma$-best-responding strategies, which is in turn a subset of $(\gamma, \delta)$-best-responding strategies.  A best-responding strategy is $0$-best-responding.  A $\gamma$-best-responding strategy is $(\gamma, 0)$-best-responding.  

Our definition of $(\gamma, \delta)$-best-responding strategies includes, for example, two other models in the literature that also relax the receiver's rationality in the Bayesian persuasion problem: 
\yc{Careful with the term bounded rationality. I'm not sure that these models are bounded rationality models.}
\tao{rephrased}
the quantal response model (proposed by \citet{mckelvey_quantal_1995} in normal-form games, recently studied by \citet{feng2024rationality} in Bayesian persuasion) and a model where the receiver makes mistakes in Bayesian update \citep{de_clippel_non-bayesian_2022}. 
\begin{example} \label{example:approximately-best-responding}
The following receiver strategies are $(\gamma, \delta)$-best-responding with some $\gamma, \delta$: 
\begin{itemize} 
	\item Quantal response: given signal $s\in S$, the receiver chooses action $a\in A$ with probability $\frac{\exp(\lambda \ru(a, \mu_s))}{\sum_{a'\in A} \exp(\lambda \ru(a', \mu_s))}$, where $\lambda \ge 0$ is a parameter.  This strategy is $(\frac{\log(|A| \lambda)}{\lambda}, \frac{1}{\lambda})$-best-responding. 
	\item Inaccurate posterior: given signal $s\in S$, the receiver forms some posterior $\mu_s'$ that is different yet close to $\mu_s$ in total variation distance $\dTV(\mu_s', \mu_s) \le \eps$.  The receiver picks a best-responding action with respect to $\mu_s'$.  This strategy is $(2\eps, 0)$-best-responding. 
\end{itemize} 
\end{example}
\noindent See Appendix~\ref{app:example:approximately-best-responding} for a proof of this example. 

\paragraph{\bf Sender's objective}
We denote by $\esu_\prior(\pi, \rho)$ the expected utility of the sender when the sender uses signaling scheme $\pi$ and the receiver uses strategy $\rho$, with prior $\prior$: 
\begin{equation*}
    \esu_\prior(\pi, \rho) = \E[\su(a, \omega)] =  \sum_{\omega \in \Omega} \prior(\omega) \sum_{s \in S} \pi(s | \omega) \E_{a \sim \rho(s)}[\su(a, \omega)] = \sum_{(\omega, s)\in \Omega \times S} \pi(\omega, s) \E_{a \sim \rho(s)}[\su(a, \omega)].
\end{equation*}
Alternatively, we can write 
\begin{equation*}
    \esu_\prior(\pi, \rho) = \sum_{s\in S} \pi(s) \E_{a \sim \rho(s)}\Big[\sum_{\omega\in \Omega} \pi(\omega | s) \su(a, \omega)\Big] = \sum_{s\in S} \pi(s) \E_{a \sim \rho(s)}[\su(a, \mu_s)]. 
\end{equation*}
Our work has two goals.
\yc{Consider rephrasing this sentence. The second is not a goal of the sender.}
\tao{Changed it to ``our'' goal.}
The first goal is to find a signaling scheme to maximize the sender's expected utility even if the receiver is allowed to take the \emph{worst} approxiamtely-best-responding strategy for the sender:
\begin{equation}
    \underline{\OBJ}(\prior, \gamma, \delta) = \sup_{\pi} \inf_{\rho: (\gamma, \delta)\text{-best-responding}} \esu_\prior(\pi, \rho),
\end{equation}
This is a ``maximin'' objective and can be regarded as a ``robust Bayesian persuasion'' problem.
The second goal is to find an upper bound on the sender's expected utility if the receiver takes the \emph{best} approximately-best-responding strategy for the sender:  
\begin{equation}
    \overline{\OBJ}(\prior, \gamma, \delta) = \sup_{\pi} \sup_{\rho: (\gamma, \delta)\text{-best-responding}} \esu_\prior(\pi, \rho). 
\end{equation}
This is a ``maximax'' objective that quantifies the maximal extent to which the sender can exploit the receiver's irrational behavior.  Clearly, $\underline{\OBJ}(\prior, \gamma, \delta) \le \overline{\OBJ}(\prior, \gamma, \delta)$.

\paragraph{\bf Restrict attentions to $\gamma$-best-responding strategies}
We will regard $\gamma, \delta$ as parameters that are close to $0$.  The following lemma shows that a small $\delta$ (the probability with which the receiver does not take a $\gamma$-best-responding action) does not affect the sender's objective a lot.  As a result, we can restrict attentions to $\gamma$-best-responding strategies for technical convenience.  
\begin{lemma}\label{lem:restrict_to_gamma-best}
	Fix a signaling scheme $\pi$.  For any $(\gamma, \delta)$-best-responding receiver strategy $\rho$, there exists a $\gamma$-best-responding receiver strategy $\tilde \rho$ such that $| \esu_\prior(\pi, \rho) - \esu_\prior(\pi, \tilde \rho) | \le \delta$. 
\end{lemma}
This lemma is proved in Appendix~\ref{app:lem:restrict_to_gamma-best}.

\paragraph{\bf Direct-revelation schemes, obedient strategy, and ``advantage''}
We define some notions that will be used in the proofs. 
A signaling scheme $\pi$ is called \emph{a direct-revelation scheme} if the signal space $S$ is equal to the action space $A$ and a sent signal $s\in S = A$ is interpreted as an action recommended to the receiver to take.  Under a direct-revelation scheme, the receiver may use the \emph{obedient strategy} $\rho^{\mathrm{ob}}$, which simply follows the recommendation from the sender: $\rho^{\mathrm{ob}}(s) = s$ with probability $1$.  The obedient strategy needs not be a best-responding strategy for the receiver.   We define the \emph{advantage of a direct-revelation scheme} as the extent to which the obedient strategy is best responding under that scheme:
\begin{definition}
	Let $\pi$ be a direct-revelation scheme.  The \emph{advantage of signal $s$ under $\pi$} is the minimal difference between the receiver's expected utilities of taking the recommended action $s$ and taking any other action $a\ne s$ given signal $s$: 
	\begin{equation}
	   \adv_\pi(s) \coloneqq \min_{a\ne s} \sum_{\omega \in \Omega} \pi(\omega | s) \Big( \ru(s, \omega) - \ru(a, \omega) \Big) = \ru(s, \mu_s) - \max_{a\neq s} \ru(a, \mu_s). 
	\end{equation}
	The \emph{advantage of $\pi$} is the minimal advantage over all signals: $ \adv_\pi \coloneqq \min_{s\in S} \adv_\pi(s)$.  
\end{definition}

The obedient strategy $\rho^\ob$ is best-responding under a direct-revelation scheme $\pi$ iff $\adv_\pi \ge 0$, and $\gamma$-best-responding iff $\adv_\pi \ge -\gamma$. \yc{Is it the reason to introducing the notion of advantage? We may want to explain why we define advantage.}
\tao{Maybe we can move the definition of advantage to Section 3.1 when we define robustification.  Basically, robustification is about increasing the advantage of a signaling scheme.  We can still keep the definitions of direct-revelation scheme and obedient stratege here. }

\section{Main Result: Persuading an Approximately Best Responding Agent}
\label{sec:main-result}
Our main results are twofold.
Under natural assumptions, we show: 
(1) a lower bound on $\underline{\OBJ}(\prior, \gamma, \delta)$: there exists a signaling scheme that guarantees the sender an expected utility almost as good as its optimal utility in the classic model, no matter what approximately best-responding strategy the receiver uses;
(2) an upper bound on $\overline{\OBJ}(\prior, \gamma, \delta)$: there is no signaling scheme that gives the sender much more utility than its optimal utility in the classic model, even if the receiver uses the approximately best-responding strategy that is best for the sender.
Both results are proved using the idea of robustification of a Bayesian persuasion scheme.

We make the following assumption on the receiver's utility:  
\begin{assumption}\label{assump:unique_optimal_action_strong}
The receiver's utility function $\ru(\cdot, \cdot)$ satisfies the following: 
\begin{itemize} 
\item For every state $\omega \in \Omega$, there exists a \emph{unique} optimal action $a_\omega = \argmax_{a\in A} \ru(a, \omega)$.  Let $\Delta$ be the minimal gap between the utilities of the optimal action and any sub-optimal action: 
\begin{equation}
    \Delta \coloneqq \min_{\omega \in \Omega, a' \ne a_\omega} \big\{ \ru(a_\omega, \omega) - \ru(a', \omega) \big\} \,>\, 0. 
\end{equation}
\item For every action $a\in A$, there exists a state for which $a$ is the unique optimal action.  Let $\Omega_a \coloneqq \{\omega\in \Omega: a_\omega = a\} \ne \emptyset$ denote the set of states for which $a$ is the unique optimal action.  
\end{itemize} 
\end{assumption}
\noindent The second part of Assumption~\ref{assump:unique_optimal_action_strong} in particular implies that the number of states is greater than or equal to the number of actions. 

Let $\OPT^\BP(\prior)$ be the sender's expected utility in the optimal Bayesian persuasion scheme under prior $\prior$ in the classic setting \citep{kamenica_bayesian_2011}, which assumes that the receiver always best responds and breaks ties in favor of the sender.  In our notation,
\[ \OPT^\BP(\prior) = \overline{\OBJ}(\prior, 0, 0) = \sup_{\pi} \sup_{\rho: \text{best-responding}} \esu_\prior(\pi, \rho).\]

\begin{theorem}\label{thm:BP-robustness}
Assume Assumption \ref{assump:unique_optimal_action_strong}. 
 Suppose $\prior_{\min} \coloneqq \min_{\omega \in \Omega} \prior(\omega) > 0$ and $\frac{\gamma}{\prior_{\min} \Delta} < 1$.  Then, 
\begin{equation}\label{eq:main-result}
\OPT^\BP(\prior) - \tfrac{\gamma}{\prior_{\min} \Delta} - \delta ~ \le ~ \underline{\OBJ}(\prior, \gamma, \delta) ~ \le ~ \overline{\OBJ}(\prior, \gamma, \delta)  ~ \le ~ \OPT^\BP(\prior) + \tfrac{\gamma}{\prior_{\min} \Delta} + \delta. 
\end{equation}
\end{theorem}

We remark that Theorem~\ref{thm:BP-robustness} holds for all bounded utility functions $\su(\cdot, \cdot)$ for the sender. 

\paragraph{Overview of the Proof of Theorem \ref{thm:BP-robustness}}
According to Lemma~\ref{lem:restrict_to_gamma-best}, we only need to prove the theorem for $\gamma$-best-responding strategies (where $\delta = 0$): 
\begin{equation}\label{eq:main-result-gamma-best-responding}
\OPT^\BP(\prior) - \tfrac{\gamma}{\prior_{\min} \Delta} ~ \le ~ \underline{\OBJ}(\prior, \gamma, 0) ~ \le ~ \overline{\OBJ}(\prior, \gamma, 0) ~ \le ~ \OPT^\BP(\prior) + \tfrac{\gamma}{\prior_{\min} \Delta}. 
\end{equation}
\begin{claim}\label{claim:implies}
\eqref{eq:main-result-gamma-best-responding} implies \eqref{eq:main-result}.
\end{claim}
\noindent This claim is proved in Appendix~\ref{app:proof_claims:implies}.

Equation \eqref{eq:main-result-gamma-best-responding} is then proved by the following two lemmas, which proves Theorem~\ref{thm:BP-robustness}.  

\begin{lemma}[Lower bound]\label{lem:lower_bound}
	$\underline{\OBJ}(\prior, \gamma, 0) \ge \OPT^\BP(\prior) - \tfrac{\gamma}{\prior_{\min} \Delta}$.  
\end{lemma}

\begin{lemma}[Upper bound]\label{lem:upper_bound}
$\overline{\OBJ}(\prior, \gamma, 0) \le \OPT^\BP(\prior) + \tfrac{\gamma}{\prior_{\min} \Delta}$.
\end{lemma}

The proofs of both Lemma~\ref{lem:lower_bound} and Lemma~\ref{lem:upper_bound} use the idea of \emph{robustification of a direct-revelation signaling scheme}, which is introduced in Section~\ref{sec:robustification}.  Section~\ref{sec:proof_lower_bound} and~\ref{sec:proof_upper_bound} then prove the two lemmas. 
We overview the main ideas here.
Roughly speaking, given any direct-revelation scheme $\pi$ under which the receiver \emph{weakly} prefers the obedient strategy $\rho^\ob$ to any other strategy $\rho$, we can turn $\pi$ into a new scheme $\pi'$ under which the receiver \emph{strictly} prefers the obedient strategy to any other strategy.  This guarantees that the receiver will still use the obedient strategy even if the receiver only approximately best responds.  Meanwhile, $\pi'$ and $\pi$ are ``close'' to each other, so the sender obtains similar expected utilities under $\pi'$ and $\pi$, with the receiver using the obedient strategy under both schemes.  By robustifying the optimal direct-revelation scheme $\pi^*$ in the classic setting, we obtain a scheme $\pi'$ that guarantees the sender an expected utility close to the expected utility of $\pi^*$, which is equal to the classic benchmark $\OPT^\BP(\prior)$; this proves the lower bound result (Lemma~\ref{lem:lower_bound}).  

To prove the upper bound result (Lemma~\ref{lem:upper_bound}),
we first turn any signaling scheme $\pi$ under which the receiver $\gamma$-best responds into a direct-revelation scheme $\pi^{\mathrm{direct}}$ under which the obedient strategy is a $\gamma$-best-responding strategy for the receiver, without decreasing the sender's utility in this step.
Then, we robustify $\pi^{\mathrm{direct}}$ into a new direct-revelation scheme $\pi'$ under which the obedient strategy is exactly best-responding.  Because the optimal scheme $\pi^*$ in the classic setting also guarantees that the obedient strategy is exactly best-responding, $\pi'$ cannot give the sender more utility than $\pi^*$ does.  Since $\pi^{\mathrm{direct}}$ and $\pi'$ are close, the sender's utility under $\pi^{\mathrm{direct}}$ thus cannot be much higher than that of $\pi^*$, which is the classic benchmark $\OPT^\BP(\prior)$.  This in turn implies that the sender's utility under $\pi$ cannot be much higher than $\OPT^\BP(\prior)$.  So the result is proved.

\subsection{Main Idea: Robustification}\label{sec:robustification}
Let $\pi$ be a direct-revelation signaling scheme. 
Let $0\le \alpha\le1$ be a parameter. 
We argue that, we can construct a new direct-revelation signaling scheme $\pi'$ from $\pi$ such that the advantages of all signals are improved, and the sender's expected utilities under $\pi'$ and $\pi$ are $\alpha$-closed.  We construct $\pi'$ as follows: 
For every state $\omega \in \Omega$, let the conditional distribution over signals $\pi'(\cdot | \omega)$ be: 
\begin{equation}\label{eq:construct_new_pi}
	\pi'(s | \omega) = (1-\alpha) \pi(s|\omega) + \alpha \mathbb{I}[s = a_\omega], \quad\quad \forall s\in S=A   
\end{equation}
($a_\omega$ was defined in Assumption~\ref{assump:unique_optimal_action_strong}).
In words, $\pi'$ does the following: given state $\omega$, with probability $1-\alpha$ send a signal (action recommendation) according to the scheme $\pi$ and with probability $\alpha$ recommend the action $s = a_\omega$ that is best for the receiver at that state. 
We note that $\pi'(\cdot | \omega)$ is a valid distribution since $\sum_{s\in S} \pi'(s | \omega) = (1-\alpha)\sum_{s\in S} \pi(s | \omega) + \alpha = 1$.
\begin{definition}[Robustification]
\label{def:robustification}
We call the $\pi'$ constructed in \eqref{eq:construct_new_pi} the \emph{$\alpha$-robustification of $\pi$}. 
\end{definition}

\begin{lemma}\label{lem:robustification}
	The $\alpha$-robustification $\pi'$ of $\pi$ satisfies the following: 
	\begin{itemize} 
	\item The marginal probability of signal $s$ under $\pi'$ is
	\begin{equation}
		\pi'(s) = (1-\alpha) \pi(s) + \alpha \prior(\Omega_s). 
	\end{equation}
	\item \label{claim:improve_adv} 
	The advantage of every signal $s\in S$ under $\pi'$ satisfies 
	\begin{equation}
		\adv_{\pi'}(s) \ge (1-\alpha) \frac{\pi(s)}{\pi'(s)} \adv_\pi(s) + \alpha \frac{\prior(\Omega_s)}{\pi'(s)}\Delta. 
	\end{equation} 
	\item  \label{claim:utility_difference}
		The sender's expected utilities under $\pi'$ and $\pi$ (assuming that the receiver uses the obedient strategy) satisfy
		\begin{equation}
			\big| \esu_{\prior}(\pi', \rho^{\mathrm{ob}}) -  \esu_{\prior}(\pi, \rho^{\mathrm{ob}}) \big| \le \alpha.
		\end{equation}  
	\end{itemize} 
\end{lemma}
\begin{proof}
	By definition, the marginal probability $\pi'(s)$ is equal to
	\begin{align*}
		\pi'(s) = \sum_{\omega\in \Omega} \prior(\omega) \pi'(s | \omega) & = \sum_{\omega\in\Omega} \prior(\omega) \Big( (1-\alpha) \pi(s|\omega) + \alpha \mathbb{I}[s = a_\omega]\Big) \\
		& = (1-\alpha) \sum_{\omega\in\Omega}\prior(\omega)\pi(s|\omega) + \alpha \sum_{\omega \in \Omega} \prior(\omega) \mathbb{I}[s = a_\omega] \\
		&  = (1-\alpha) \pi(s) + \alpha \sum_{\omega \in \Omega_s} \prior(\omega)  ~~ = ~~ (1-\alpha) \pi(s) + \alpha \prior(\Omega_s).  
	\end{align*}
	
	To analyze the advantage $\adv_{\pi'}(s) = \min_{a\ne s} \sum_{\omega \in \Omega} \pi'(\omega | s) \big(\ru(s, \omega) - \ru(a, \omega) \big)$, we first write the posterior probability of a state $\omega \in \Omega$ of states given signal $s\in S$.
	\begin{align*}
		\pi'(\omega | s)  = \frac{\prior(\omega) \pi'(s|\omega)}{\pi'(s)} & = \frac{\prior(\omega) \Big( (1-\alpha) \pi(s|\omega) + \alpha \mathbb{I}[s = a_\omega] \Big)}{\pi(s)}\frac{\pi(s)}{\pi'(s)} \\
		& = (1-\alpha) \frac{\pi(s)}{\pi'(s)} \frac{\prior(\omega) \pi(s|\omega)}{\pi(s)} \,+\, \alpha \frac{\prior(\omega)}{\pi'(s)} \mathbb{I}[s = a_\omega]  \\
		& = (1-\alpha) \frac{\pi(s)}{\pi'(s)} \pi(\omega | s) \,+\, \alpha \frac{\prior(\omega)}{\pi'(s)} \mathbb{I}[s = a_\omega]. 
	\end{align*} 
	So, for any $a\ne s$,
	\begin{align*}
		& \sum_{\omega \in \Omega} \pi'(\omega | s) \Big(\ru(s, \omega) - \ru(a, \omega) \Big) \\
		& = (1-\alpha)\frac{\pi(s)}{\pi'(s)} \sum_{\omega\in\Omega}\pi(\omega|s) \Big(\ru(s, \omega) - \ru(a, \omega) \Big) \,+\, \alpha \frac{1}{\pi'(s)} \sum_{\omega \in \Omega_s} \prior(\omega) \Big(\ru(s, \omega) - \ru(a, \omega) \Big). 
	\end{align*}
	We note that $\sum_{\omega\in\Omega}\pi(\omega|s) \big(\ru(s, \omega) - \ru(a, \omega) \big) \ge \adv_\pi(s)$ by definition.  And recall from Assumption \ref{assump:unique_optimal_action_strong} that $\ru(s, \omega) - \ru(a, \omega) \ge \Delta$ for $\omega\in\Omega_s, a\ne s$. Therefore, we have 
	\begin{align*}
	\sum_{\omega \in \Omega} \pi'(\omega | s) \Big(\ru(s, \omega) - \ru(a, \omega) \Big) & \ge  (1-\alpha)\frac{\pi(s)}{\pi'(s)} \adv_\pi(s) ~ + ~ \alpha \frac{1}{\pi'(s)} \sum_{\omega \in \Omega_s} \prior(\omega) \Delta \\
	& = (1-\alpha)\frac{\pi(s)}{\pi'(s)} \adv_\pi(s) ~ + ~ \alpha \frac{\prior(\Omega_s)}{\pi'(s)} \Delta.
	\end{align*}
	Minimizing over $a\ne s$, we obtain $\adv_{\pi'}(s) \ge (1-\alpha)\frac{\pi(s)}{\pi'(s)} \adv_\pi(s) ~ + ~ \alpha \frac{\prior(\Omega_s)}{\pi'(s)} \Delta$.

	Finally, we prove the third item (sender's utility). 
	Consider the total variation distance between the two joint distributions of states and signals, $\pi'(\cdot, \cdot)$ and $\pi(\cdot, \cdot)$: 
	\begin{align*}
		\dTV(\pi', \pi) & = \frac{1}{2} \sum_{\omega, s} \Big| \pi'(\omega, s)  - \pi(\omega, s) \Big| 
		~=~ \frac{1}{2} \sum_{\omega} \prior(\omega) \sum_s \Big| \pi'(s|\omega)  - \pi(s|\omega) \Big|  \\
		& = \frac{1}{2} \sum_{\omega} \prior(\omega) \sum_s \Big| -\alpha \pi(s|\omega) + \alpha \mathbb{I}[s=a_\omega] \Big| \\
		& \le \frac{1}{2} \sum_{\omega} \prior(\omega) \Big( \sum_s \alpha \pi(s|\omega) + \alpha \Big) ~ = ~ \frac{1}{2} \sum_{\omega} \prior(\omega) 2\alpha ~ = ~ \alpha. 
	\end{align*}
	Since the utility $\su(s, \omega)$ is bounded in $[0, 1]$, the difference in expected utilities is bounded by: 
	\begin{align*}
		\Big| \esu_{\prior}(\pi', \rho^{\mathrm{ob}}) -  \esu_{\prior}(\pi, \rho^{\mathrm{ob}}) \Big| 
		 \, =\, \Big| \sum_{\omega, s} \pi'(\omega, s) \su(s, \omega)  \,-\, \sum_{\omega, s} \pi(\omega, s) \su(s, \omega) \Big| \,\le\,  \dTV(\pi', \pi)\cdot 1  \,\le\, \alpha, 
	\end{align*}
	as claimed. 
\end{proof}

\subsection{Lower Bounding: Proof of Lemma \ref{lem:lower_bound}}
\label{sec:proof_lower_bound}
Let $\pi^*$ be an optimal signaling scheme for the sender under prior $\prior$ in the classic Bayesian persuasion setting.  As shown by \cite{kamenica_bayesian_2011}, $\pi^*$ can be restricted to be a direct-revelation scheme.  In the classic setting, the receiver's obedient strategy $\rho^{\mathrm{ob}}$ is a best-responding strategy, which implies that the advantage $\adv_{\pi^*}(s) \ge 0$ for every $s\in S=A$.
We let $\alpha$ be an arbitrarily small real number that is greater than  $\frac{\gamma}{\prior_{\min}\Delta}$ and let $\pi'$ be the $\alpha$-robustification of $\pi^*$ (Definition \ref{def:robustification}).
According to Lemma~\ref{claim:improve_adv}, the advantage $\adv_{\pi'}(s)$ satisfies: 
\begin{align*}
	\adv_{\pi'}(s) \,\ge\, (1-\alpha) \frac{\pi^*(s)}{\pi'(s)} \adv_{\pi^*}(s) + \alpha \frac{\prior(\Omega_s)}{\pi'(s)}\Delta  \,\ge\, 0 + \alpha \frac{\prior(\Omega_s)}{\pi'(s)}\Delta \,=\, \alpha \frac{\prior(\Omega_s)}{\pi'(s)}\Delta \,\ge\, \alpha \prior_{\min}\Delta,
\end{align*}
where the last inequality is because $\prior(\Omega_s) \ge \prior_{\min}$ and $\pi'(s) \le 1$.
So, when $\alpha > \frac{\gamma}{\prior_{\min}\Delta}$ we have $\adv_{\pi'}(s) > \gamma$. 
This implies that the action $s$ is the \emph{unique} $\gamma$-best-responding action for the receiver given signal $s$ under signaling scheme $\pi'$, so the obedient strategy $\rho^{\mathrm{ob}}$ is the only $\gamma$-best responding strategy for the receiver.  Therefore, we have
\begin{align*}
	\underline{\OBJ}(\prior, \gamma, 0) \,\ge\, \inf_{\rho: \gamma\text{-best-responding}} \esu_\prior(\pi', \rho) \,=\, \esu_\prior(\pi', \rho^{\mathrm{ob}}) \,\ge\, \esu_\prior(\pi^*, \rho^{\mathrm{ob}}) - \alpha \,= \, \OPT^\BP(\prior) - \alpha, 
\end{align*}
where the second ``$\ge$'' is due to Lemma~\ref{claim:utility_difference}.  Letting $\alpha \to \frac{\gamma}{\prior_{\min}\Delta}$ proves the lemma.

\subsection{Upper Bounding: Proof of Lemma \ref{lem:upper_bound}}
\label{sec:proof_upper_bound}
Let $\pi$ be any signaling scheme and $\rho$ be any $\gamma$-best-responding receiver strategy.
We want to prove that the sender's expected utility satisfies $\esu_\prior(\pi, \rho) \le \OPT^{\BP}(\prior) + \frac{\gamma}{\mu_{\min}\Delta}$, so that $\overline{\OBJ}(\prior, \gamma, 0) = \sup_\pi \sup_{\rho: \gamma\text{-best-responding}} \esu_\prior(\pi, \rho) \le \OPT^{\BP}(\prior) + \frac{\gamma}{\mu_{\min}\Delta}$ is proved.

Recall that $A_\pi^\gamma(s)$ is the set of $\gamma$-best-responding actions to signal $s$ for the receiver, under signaling scheme $\pi$.  Let $a^*(s) = \argmax_{a\in A_\pi^\gamma(s)} \sum_{\omega \in \Omega} \pi(\omega | s) \su(a, \omega)$ be a $\gamma$-best-responding action of the receiver that maximizes the \emph{sender}'s expected utility, given signal $s$.  Let $\rho^*$ be the receiver strategy that takes action $a^*(s)$ deterministically at every signal $s \in S$.  Since $a^*(s)$ maximizes the sender's expected utility, we have
\begin{align}
\esu_\prior(\pi, \rho) & ~ =~  \sum_{s\in S} \pi(s) \sum_{\omega \in \Omega} \pi(\omega | s) \E_{a\sim \rho(s)}\big[ \su(a, \omega) \big] \\
& ~ \le~  \sum_{s\in S} \pi(s) \sum_{\omega \in \Omega} \pi(\omega | s) \su(a^*(s), \omega) ~ =~  \esu_\prior(\pi, \rho^*).  \label{eq:upperbound-step-1}
\end{align}
We also note that $\rho^*$ is a $\gamma$-best-responding strategy by definition. 

We then show that, for any signaling scheme $\pi$ and any deterministic (i.e., $\rho(s)$ is equal to some action $a\in A$ with probability $1$ for every $s\in S$) $\gamma$-best-responding receiver strategy $\rho^{\mathrm{det}}$, there exists a corresponding direct-revelation scheme $\pi^{\text{direct}}$ for which the obedient strategy $\rho^{\ob}$ is $\gamma$-best-responding and the sender obtains the same expected utility under $(\pi, \rho^{\mathrm{det}})$ and $(\pi^{\mathrm{direct}}, \rho^{\ob})$.
The idea is simple: we let the signaling scheme $\pi^{\mathrm{direct}}$ recommend action $a$ whenever the scheme $\pi$ sends a signal $s$ for which the receiver takes actions $a$.  Formally, let $S_a = \{s \in S : \rho^{\mathrm{det}}(s) = a \}$ be the set of signals for which the receiver takes action $a$.  Let
\begin{equation}\label{eq:pi-direct-definition}
    \pi^{\mathrm{direct}}(a \smid \omega) = \pi(S_a \smid \omega) = \sum_{s\in S_a} \pi(s \smid \omega), \quad \quad \forall \omega \in \Omega, a\in A. 
\end{equation}
\begin{claim}\label{claim:deterministic-direct}
The $\pi^{\mathrm{direct}}$ constructed above satisfies: (1) $\rho^\ob$ is a $\gamma$-best-responding strategy for the receiver; (2) the sender's expected utility $\esu_\prior(\pi, \rho^{\mathrm{det}}) = \esu_\prior(\pi^{\mathrm{direct}}, \rho^{\ob})$.
\end{claim}
\noindent The proof of this claim is in Appendix~\ref{app:claim:deterministic-direct}. 

Since $\rho^*$ is a deterministic $\gamma$-best-responding receiver strategy, Claim~\ref{claim:deterministic-direct} implies that there exists a direct-revelation scheme $\pi^{\mathrm{direct}}$ such that
\begin{equation}  \label{eq:rho-star-direct}
\esu_\prior(\pi, \rho^*) = \esu_\prior(\pi^{\mathrm{direct}}, \rho^{\ob})
\end{equation} 
and $\rho^\ob$ is a $\gamma$-best-responding strategy for the receiver.  Then, we use the robustification idea in Section \ref{sec:robustification}: we robustify the scheme $\pi^{\mathrm{direct}}$ into a new direct-revelation scheme $\pi'$ for which $\rho^\ob$ is an exactly best-responding strategy for the receiver (where $\gamma = 0$).  In particular, according to Lemma~\ref{lem:robustification}, the $\alpha$-robustification scheme $\pi'$ of $\pi^{\mathrm{direct}}$ has advantage
\begin{align}
\adv_{\pi'}(s) \ge (1-\alpha) \frac{\pi^{\mathrm{direct}}(s)}{\pi'(s)} \adv_{\pi^{\mathrm{direct}}}(s) + \alpha \frac{\mu(\Omega_s)}{\pi'(s)}\Delta, \quad\quad  \forall s \in S = A. 
\end{align}
Since $\rho^\ob$ is $\gamma$-best-responding with respect to $\pi^{\mathrm{direct}}$, we have $\adv_{\pi^{\mathrm{direct}}}(s) \ge -\gamma$ by definition.  So, 
\begin{align}
\adv_{\pi'}(s) \ge - (1-\alpha) \frac{\pi^{\mathrm{direct}}(s)}{\pi'(s)} \gamma + \alpha \frac{\mu(\Omega_s)}{\pi'(s)}\Delta. 
\end{align}
To make sure $\rho^\ob$ is an exactly best-responding strategy with respect to $\pi'$, we need $\adv_{\pi'}(s) \ge 0$, and this is satisfied when 
\begin{align}
\alpha \ge \frac{\frac{\pi^{\mathrm{direct}}(s)}{\pi'(s)}\gamma}{\frac{\pi^{\mathrm{direct}}(s)}{\pi'(s)}\gamma + \frac{\mu(\Omega_s)}{\pi'(s)}\Delta} = \frac{\pi^{\mathrm{direct}}(s)\gamma}{\pi^{\mathrm{direct}}(s)\gamma + \mu(\Omega_s)\Delta} = \frac{\gamma}{\gamma + \frac{\mu(\Omega_s)}{\pi^{\mathrm{direct}}(s)}\Delta}. 
\end{align}
Since $\frac{\gamma}{\gamma + \frac{\mu(\Omega_s)}{\pi^{\mathrm{direct}}(s)}\Delta} \le \frac{\gamma}{\gamma + \mu(\Omega_s)\Delta} \le \frac{\gamma}{\mu(\Omega_s)\Delta} \le \frac{\gamma}{\mu_{\min}\Delta}$, it suffices to let $\alpha = \frac{\gamma}{\mu_{\min} \Delta}$.  With this choice of $\alpha$, by Lemma~\ref{lem:robustification} we have 
\begin{align}\label{eq:pi-prime-plus-alpha}
    \esu_\prior(\pi^{\mathrm{direct}}, \rho^\ob) \le \esu_\prior(\pi', \rho^\ob) + \alpha.
\end{align}
Since $\rho^\ob$ is best-responding with respect to $\pi'$, we have 
\begin{align}\label{eq:pi-prime-less-than-OPT}
     \esu_\prior(\pi', \rho^\ob) \,\le\, \sup_{\pi} \sup_{\rho:\text{best-responding}} \esu_\prior(\pi, \rho) \,=\, \OPT^{\BP}(\mu). 
\end{align}
Thus, we obtain the following chain of inequalities,  
\begin{align*}
  \esu_\prior(\pi, \rho) \stackrel{\eqref{eq:upperbound-step-1}}{\le} \esu_\prior(\pi, \rho^*) \stackrel{\eqref{eq:rho-star-direct}}{=} \esu_\prior(\pi^{\mathrm{direct}}, \rho^{\ob}) \stackrel{\eqref{eq:pi-prime-plus-alpha}}{\le} \esu_\prior(\pi', \rho^\ob) + \tfrac{\gamma}{\mu_{\min}\Delta}  \stackrel{\eqref{eq:pi-prime-less-than-OPT}}{\le} \OPT^{\BP}(\mu) + \tfrac{\gamma}{\mu_{\min}\Delta}. 
\end{align*}
which concludes the proof.

\section{Application: Persuading a Learning Agent}
\label{sec:learning}
In this section, we consider a repeated Bayesian persuasion model where the receiver, or agent, learns to respond to the sender's signals using some algorithms.  We apply our techniques in the approximately-best-responding model (in particular, the robustification idea) to this model to show a similar conclusion to the lower bound in Theorem~\ref{thm:BP-robustness}: with a learning receiver, the sender can always do almost at least as well as in the classic non-learning model.  And interestingly, the upper bound result in Theorem~\ref{thm:BP-robustness} does not apply here: we find an example where the sender can do much better in the learning model than in the classic model. 

\paragraph{\bf Model}
As in the model in Section~\ref{sec:model}, we have fixed state space $\Omega$, signal space $S$, action space $A$, prior $\prior$, sender utility $\su(\cdot, \cdot)$, and receiver utility $\ru(\cdot, \cdot)$.
The sender and the receiver interact for infinite rounds.\footnote{We consider an infinite horizon model only for ease of presentation.  Our qualitative results also apply to the finite horizon model where $1\le t \le T < \infty$ for a large $T$.}  
At each round $t = 1, 2, \ldots$, the following events happen in sequence: 
\begin{enumerate} 
    \item The sender chooses some signaling scheme $\pi^{(t)} : \Omega \to \Delta(S)$. 
    \item A fresh state of the world $\omega^{(t)}$ is drawn according to the fixed prior distribution $\prior$.
    \item The sender observes $\omega^{(t)}$ and sends a (randomized) signal $s^{(t)} \sim \pi^{(t)}(\omega^{(t)})$ to the receiver. 
    \item Upon receiving $s^{(t)}$, the receiver takes some (randomized) action $a^{(t)}\in A$. 
    \item The sender obtains utility $\su^{(t)} = \su(a^{(t)}, \omega^{(t)})$ and the receiver obtains utility $\ru^{(t)} = \ru(a^{(t)}, \omega^{(t)})$. 
\end{enumerate}

We discuss what information the sender and the receiver have when making decisions. 
For the sender, when choosing the signaling scheme $\pi^{(t)}$ at step (1), the sender knows the past states, signals, actions, and signaling schemes: $(\omega^{(\tau)}, s^{(\tau)}, a^{(\tau)}, \pi^{(\tau)})_{\tau = 1}^{t-1}$.
For the receiver, when choosing the action $a^{(t)}$ at step (4), besides the current signal $s^{(t)}$ the receiver knows the past signals, actions, and its own utilities: $(s^{(\tau)}, a^{(\tau)}, \ru^{(\tau)})_{\tau = 1}^{t-1}$.
Crucially, the receiver does \emph{not} know the current signaling scheme $\pi^{(t)}$; this assumption is different from Section~\ref{sec:model}.
In addition, there are two feedback models regarding whether the receiver is able to observe the state of the world $\omega^{(t)}$ \emph{in the end of} each round: 
With \emph{full-information feedback}, the receiver observes the state of the world and is able to compute the utility $\ru(a, \omega^{(t)})$ it would have obtained if it took action $a$ this round, for all $a\in A$.
With \emph{partial-information feedback}, the receiver does not observe the state of the world and only knows the utility $\ru^{(t)} = \ru(a^{(t)}, \omega^{(t)})$ of the actually taken action $a^{(t)}$.
Mathematically, with full-information feedback the receiver knows $s^{(t)}, (s^{(\tau)}, a^{(\tau)}, \omega^{(\tau)})_{\tau = 1}^{t-1}$ when choosing $a^{(t)}$; with partial-information feedback the receiver knows $s^{(t)}, (s^{(\tau)}, a^{(\tau)}, \ru^{(\tau)})_{\tau = 1}^{t-1}$.
Our results apply to both feedback models. 

The receiver's learning problem can be regarded as a contextual multi-armed bandit problem \citep{tyler_lu_contextual_2010}.  $A$ is the set of arms.  The signal $s^{(t)}$ from the sender serves as a context, which affects the utility of each arm $a \in A$.  The receiver picks an arm to pull based on the current context and the historical information about the utilities of each arm under different contexts, using some learning algorithms which we introduce below.

\paragraph{\bf Learning algorithms} 
In this work, we consider a general class of learning algorithms for the receiver that are based on the idea of approximately best-responding to the history.
We first define some notations.
Let $\hat{\mu}_s^{(t)}$ be the empirical distribution of the state of the world conditioning on signal $s$ in the first $t$ rounds: $\forall \omega \in \Omega$,
\begin{equation}
  \hat{\mu}_s^{(t)}(\omega) = \frac{\sum_{\tau=1}^t \mathbbm{1}[s^{(\tau)}=s, \omega^{(\tau)} = \omega]}{\sum_{\tau=1}^t \mathbbm{1}[s^{(\tau)} = s]}. 
\end{equation}
$\hat{\mu}_s^{(t)}$ is not well-defined if $\sum_{\tau=1}^t \mathbbm{1}[s^{(\tau)} = s] = 0$.
Let $\ru(a, \hat{\mu}_s^{(t)})$ denote the receiver's expected utility of taking action $a$ on distribution $\hat{\mu}_s^{(t)}$, which is also the receiver's empirical average utility of taking action $a$ in all the first $t$ rounds with signal $s$: 
\begin{align}
    \ru(a, \hat{\mu}_s^{(t)}) = \sum_{\omega\in\Omega} \hat{\mu}_s^{(t)}(\omega) \ru(a, \omega)
    & = \sum_{\omega\in\Omega} \frac{\sum_{\tau=1}^t \mathbbm{1}[s^{(\tau)}=s, \omega^{(\tau)} = \omega] \ru(a, \omega)}{\sum_{\tau=1}^t \mathbbm{1}[s^{(\tau)} = s]} \\
    & = \frac{\sum_{\tau=1}^t \mathbbm{1}[s^{(\tau)}=s] \ru(a, \omega^{(\tau)})}{\sum_{\tau=1}^t \mathbbm{1}[s^{(\tau)} = s]}. 
\end{align}
As in Section~\ref{sec:model}, an action $a$ is $\gamma$-best-responding to $\hat \mu_s^{(t)}$ if $\ru(a, \hat{\mu}_s^{(t)}) \ge \ru(a', \hat{\mu}_s^{(t)}) - \gamma$ for all $a'\in A$. 
\begin{definition}[Empirical $(\gamma_t, \delta_t)$-best-responding algorithm]
Let $(\gamma_t)_{t\ge 1}, (\delta_t)_{t\ge 1}$ be two sequences.
An algorithm is an \emph{empirical $(\gamma_t, \delta_t)$-best-responding algorithm} if the following holds: at each round $t\ge 1$, if the signal $s^{(t)}$ is $s$ and $\hat{\mu}_s^{(t-1)}$ is well-defined, then the algorithm chooses an action $a^{(t)}$ that is $\gamma_t$-best-responding to $\hat \mu_s^{(t-1)}$ with probability at least $1 - \delta_t$.  
\end{definition}
Our definition of empirical $(\gamma_t, \delta_t)$-best-responding algorithm is inspired by and generalizes the ``mean-based learning algorithms'' in \cite{braverman_selling_2018}.  We give some examples of empirical $(\gamma_t, \delta_t)$-best-responding algorithms. 
For the finite horizon version of our model (where $1\le t \le T$ for some $T< \infty$), the contextual versions of Exponential Weights, Multiplicative Weights Update, and Follow the Perturbed Leader are empirical $\big(\gamma_t = O(\frac{\sqrt{T} \log T}{t}), \delta_t = O(\frac{1}{\sqrt T})\big)$-best-responding algorithms with full-information feedback.\footnote{By ``the contextual version of a multi-armed bandit algorithm'', we mean running the multi-armed bandit algorithm for each context separately.  See \cite{braverman_selling_2018} for a formal definition.}
The contextual version of EXP-3 is an empirical $(\gamma_t = O\big(\frac{T^{3/4}\log T}{t}), \delta_t = O(\frac{1}{T^{1/4}})\big)$-best-responding algorithm with partial-information feedback. (See \cite{braverman_selling_2018} for the proofs.)
We also give an example of an infinite horizon empirical $(\gamma_t, \delta_t)$-best-responding algorithm, which is the contextual version of the infinite horizon Exponential Weights algorithm \citep{cesa-bianchi_prediction_2006}: 

\begin{example}\label{ex:infinite-empirical-best-responding}
Let $(\eta_t)_{t\ge 1}$ be a sequence with $\eta_t = \sqrt{\frac{\log|A|}{t}}$.  
The contextual infinite horizon Exponential Weight is the following: At each round $t$, given context/signal $s^{(t)} = s$, the algorithm chooses action $a^{(t)} = a$ with probability
\begin{equation}
    \frac{\exp\big(\eta_t \sum_{\tau \le t-1, s^{(\tau)} = s} \ru(a, \omega^{(\tau)}) \big)}{\sum_{a'\in A} \exp\big(\eta_t \sum_{\tau \le t-1, s^{(\tau)} = s} \ru(a', \omega^{(\tau)}) \big)}. 
\end{equation}
If every signal $s$ is sent with probability at least $p$ at every round, then this algorithm is an empirical $(\gamma_t, \delta_t)$-best-responding algorithm with $\gamma_t = O\Big( \frac{\log(|A|t)}{p \sqrt{t \log|A|}}  \Big)$ and $\delta_t = O\Big(\frac{1}{p \sqrt{t\log|A|}} \Big)$.
\end{example}
See Appendix~\ref{app:ex:infinite-empirical-best-responding} for a proof of this example. 
We remark that the choice of $\eta_t = \sqrt{\frac{\log|A|}{t}}$ guarantees the algorithm to be no-regret (Theorem 2.3 in \cite{cesa-bianchi_prediction_2006}).

\subsection{The sender can do almost as well as the classic model}
\begin{theorem}\label{thm:learning}
Assume Assumption~\ref{assump:unique_optimal_action_strong} and suppose $\mu_{\min} := \min_{\omega\in\Omega}\mu(\omega) > 0$.
Suppose the receiver uses an empirical $(\gamma_t, \delta_t)$-best-responding algorithm, with $\gamma_t\to 0$ and $\delta_t\to 0$ as $t\to\infty$.
For any constant $C>0$, there exists a signaling scheme $\pi$ such that, by using $\pi$ for all rounds, the sender can obtain an expected average utility $\E[\frac{1}{t}\sum_{\tau=1}^t \su(a^{(\tau)}, \omega^{(\tau)})]$ of at least $\OPT^\BP(\prior) - C$ as $t\to\infty$.
\end{theorem}

We remark that, in order to find the signaling scheme $\pi$ above, the sender does not need to know the exact algorithm of the receiver; knowing an upper bound on the $(\gamma_t, \delta_t)$ sequence suffices.

\begin{proof}[Proof of Theorem~\ref{thm:learning}]
The proof will use the following concentration bound, whose proof is in Appendix~\ref{app:lem:concentration}.
\begin{lemma}\label{lem:concentration}
Suppose the sender uses a fixed signaling scheme $\pi$ for all rounds. 
Suppose $\sqrt{\frac{3\log(2|S|t)}{\pi(s)t}} < \frac{1}{2}$ for all $s\in S$ with $\pi(s) > 0$.  Then, with probability at least $1 - \frac{2}{t}$, we have: for any $s\in S$ with $\pi(s) > 0$, any $a\in A$, $|\ru(a, \hat \mu_s^{(t)}) - \ru(a, \mu_s)| \le 2\sqrt{\frac{3\log(2|S|t)}{\pi(s)t}} + \frac{2}{\pi(s)}\sqrt{\frac{\log(2|S||A|t)}{2t}}$.  
\end{lemma}

Similar to the proof of Lemma~\ref{lem:lower_bound}, let $\pi^*$ be an optimal signaling scheme for the sender for prior $\prior$ in the classic setting, which can be assumed to be a direct-revelation scheme \citep{kamenica_bayesian_2011}, so $S = A$.
Let $\alpha = \frac{C}{2} > 0$ and let $\pi$ be the $\alpha$-robustification of $\pi^*$ (Definition~\ref{def:robustification}).
Suppose the sender uses the scheme $\pi$ for all rounds. 
Since $\gamma_t \to 0$ and $\delta_t\to 0$, for large enough $t$ we have 
\begin{equation}\label{eq:large-t-1}
    \alpha > \tfrac{\pi(s)}{\mu(\Omega_s)\Delta}\gamma_t + \tfrac{4}{\mu(\Omega_s)\Delta}\sqrt{\tfrac{\pi(s) \cdot 3\log(2|S|(t-1))}{t-1}} + \tfrac{4}{\mu(\Omega_s)\Delta} \sqrt{\tfrac{\log(2|S||A|(t-1))}{2(t-1)}} ~\; \text{ for every $s\in S$} 
\end{equation}
and
\begin{equation}\label{eq:large-t-2}
    \alpha + \delta_t + \tfrac{2}{t-1} < C. 
\end{equation}
We consider each such $t$.  By Lemma~\ref{lem:concentration}, with probability at least $1-\frac{2}{t-1}$, for any $s$ and $a$ we have
\begin{equation}\label{eq:c}
 |\ru(a, \hat \mu_s^{(t-1)}) - \ru(a, \mu_s)| \,\le\, 2\sqrt{\tfrac{3\log(2|S|(t-1))}{\pi(s)(t-1)}} + \tfrac{2}{\pi(s)}\sqrt{\tfrac{\log(2|S||A|(t-1))}{2(t-1)}} \, =: \, c.
\end{equation}

\begin{claim}\label{claim:learning-approxiamtely-best-response}
At round $t$, for any possible signal $s^{(t)} = s \in S$, the receiver takes a $(\gamma_t + 2c)$-best-responding action $a^{(t)}$ with respect to $\mu_s$ with probability at least $1-\delta_t-\frac{2}{t-1}$.
\end{claim}
We prove this claim.  Let $a^*$, $\hat a^*$ be the best-responding action for distribution $\mu_s$, $\hat \mu_s^{(t-1)}$, respectively.  By the assumption that the receiver uses an empirical $(\gamma_t, \delta_t)$-best-responding algorithm, with probability at least $1-\delta_t$ the action $a^{(t)}$ taken by the receiver is $\gamma_t$-best-responding to $\hat \mu_s^{(t-1)}$.  The expected utility of $a^{(t)}$ on the true distribution $\mu_s$ satisfies 
\begin{align*}
    \ru(a^{(t)}, \mu_s) & ~~~ \stackrel{\eqref{eq:c}}{\ge} ~~~ \ru(a^{(t)}, \hat \mu_s^{(t-1)}) - c  ~~ \stackrel{\text{$a^{(t)}$ is $\gamma_t$-best-responding to $\hat \mu_s^{(t-1)}$}}{\ge} ~ \ru(\hat a^*, \hat \mu_s^{(t-1)}) - \gamma_t - c  \\
    & \quad\quad  \stackrel{\text{$\hat a^*$ is best-responding to $\hat \mu_s^{(t-1)}$}}{\ge} ~ \ru(a^*, \hat \mu_s^{(t-1)}) - \gamma_t - c  \quad \stackrel{\eqref{eq:c}}{\ge} \quad \ru(a^*, \mu_s) - \gamma_t - 2c. 
\end{align*}
So, the taken action $a^{(t)}$ is $(\gamma_t + 2c)$-best-responding to $\mu_s$.  This happens with probability at least $1 - \delta_t - \frac{2}{t-1}$.  The claim is thus proved. 

According to Lemma~\ref{lem:robustification}, the signaling scheme $\pi$, as the $\alpha$-robustification of $\pi^*$, has advantage
\begin{align*}
    \adv_\pi(s) \,\ge\, (1-\alpha) \tfrac{\pi^*(s)}{\pi(s)}\adv_{\pi^*}(s) + \alpha \tfrac{\mu(\Omega_s)}{\pi(s)} \Delta  \, \ge \, \alpha \tfrac{\mu(\Omega_s)}{\pi(s)} \Delta
\end{align*}
for every $s\in S$ (the last inequality is because $\adv_{\pi^*}(s) \ge 0$).  Given \eqref{eq:large-t-1}, we have 
\begin{align*}
    \adv_\pi(s) \, \ge \, \alpha \tfrac{\mu(\Omega_s)}{\pi(s)} \Delta \, > \, \gamma_t + 4 \sqrt{\tfrac{3\log(2|S|(t-1))}{\pi(s)(t-1)}} + \tfrac{4}{\pi(s)}\sqrt{\tfrac{\log(2|S||A|(t-1))}{2(t-1)}} \, = \, \gamma_t + 2c. 
\end{align*}
This means that the recommended action $s$ is $(\gamma_t+2c)$-better than any other action and is hence the only $(\gamma_t+2c)$-best-responding action with respect to $\mu_s$.  Since Claim \ref{claim:learning-approxiamtely-best-response} shows that the receiver takes a $(\gamma_t + 2c)$-best-responding action with respect to $\mu_s$ at round $t$ (with probability at least $1-\delta_t-\frac{2}{t-1}$), this action must be $s$.  In other words, the receiver uses the obedient strategy at round $t$ with probability at least $1 - \delta_t - \frac{2}{t-1}$.  So, the sender's expected utility at round $t$ is at least
\begin{align*}
    \E[\su(a^{(t)}, \omega^{(t)})] \ge \esu_\prior(\pi, \rho^\ob) - \delta_t - \tfrac{2}{t-1} & \stackrel{\text{(Lemma~\ref{lem:robustification})}}{\ge} \esu_\prior(\pi^*, \rho^\ob) - \alpha - \delta_t - \tfrac{2}{t-1} \stackrel{\eqref{eq:large-t-2}}{>} \OPT^\BP(\prior) - C.  
\end{align*}
So, we have $\lim_{t\to\infty} \E[\su(a^{(t)}, \omega^{(t)})] \ge \OPT^\BP(\prior) - C$ in the limit. 
Taking the average, we obtain
\[ \lim_{t\to\infty} \E\Big[\frac{1}{t}\sum_{\tau=1}^t \su(a^{(\tau)}, \omega^{(\tau)})\Big] \,=\, \lim_{t\to\infty} \frac{1}{t}\sum_{\tau=1}^t \E[\su(a^{(\tau)}, \omega^{(\tau)})] \,\ge\, \OPT^\BP(\prior) - C, \]
which proves the theorem. 
\end{proof} 

\subsection{An example where the sender can do better than the classic model}
In Theorem~\ref{thm:BP-robustness} we showed that the sender cannot do much better than the classic Bayesian persuasion benchmark in the approximately-best-responding model.  Interestingly, this conclusion does not hold in the learning receiver model.  
The following example shows that, with a receiver using an empirical best-responding learning algorithm, it is possible for the sender to do much better than the classic Bayesian persuasion benchmark.

\begin{example}\label{ex:better-than-classic-model}
Suppose there are two states of the world: Good (G) and Bad (B); two actions: $a$ and $b$.  The prior is $\prior(G) = \prior(B) = 1/2$.  The utility matrices of the sender and the receiver are as follows:
\begin{center}
\begin{tabular}{|l|l|l|lll|l|l|l|}
\cline{1-3} \cline{7-9}
sender & G  & B  &  & \hspace{1em} &  & receiver & G & B \\ \cline{1-3} \cline{7-9} 
$a$ & $0$         & $1$         &  &  &  & $a$ & $1$        & $0$        \\ \cline{1-3} \cline{7-9} 
$b$ & $1$ & $0$ &  &  &  & $b$ & $0$        & $1$        \\ \cline{1-3} \cline{7-9}
\end{tabular}
\end{center}
In words, the sender wants the receiver's action to mismatch with the state while the receiver wants to match with the state.  (This example satisfies Assumption~\ref{assump:unique_optimal_action_strong}.)

The receiver uses an empirical best-responding algorithm, with $\gamma_t = \delta_t=0$, breaking ties randomly.  Formally, at any round $t$, given signal $s^{(t)} = s$, if the empirical probability of G is $\hat \mu_s^{(t-1)}(G) > 1/2$, the receiver takes action $a$; if $\hat \mu_s^{(t-1)}(G) < 1/2$, takes $b$; otherwise, randomly chooses $a$ or $b$. 

The sender uses the following time-varying signaling scheme: Continuously send signals $s_2$ until a Good state occurs (say at round $t_1$), send $s_1$ at that round.  Then, continuously send signals $s_2$ until a Bad state occurs (say at round $t_2$), send $s_1$ at that round. Keep alternating in this way: continuously send $s_2$ until a Good state occurs, send $s_1$; continuously send $s_2$ until a Bad state occurs, send $s_1$; .... 

We claim that, by using the above signaling scheme, the sender can obtain an expected average utility of $\frac{5}{8} = 0.625$ as $t\to\infty$.  The sender's optimal signaling scheme in the classic non-learning model is to reveal no information, where the receiver responds by taking either action $a$ or action $b$; the sender's expected utility in this scheme is $0.5$, smaller than $0.625$ by a constant margin. 
\end{example}
\begin{proof}
An example sequence of realized states, signals, and actions is as follows, where ``?'' denotes a randomly chosen action, and bold symbols mark the rounds where signal $s_1$ is sent. 
\begin{center}
\begin{tabular}{cccccccccccc}
    $B$   &   $\mathbf{G}$  & $G$   & $G$   & $\mathbf{B}$   & $B$   & $\mathbf{G}$   & $\mathbf{B}$   & $\mathbf{G}$   & $\mathbf{B}$   & $B$ &  ... \\
    $s_2$ & $\mathbf{s_1}$ & $s_2$ & $s_2$ & $\mathbf{s_1}$ & $s_2$ & $\mathbf{s_1}$ & $\mathbf{s_1}$ & $\mathbf{s_1}$ & $\mathbf{s_1}$ & $s_2$ & ... \\
    ?     &   {\bf ?}     & $b$   & ?     & $\mathbf{a}$   & $a$   & {\bf ?}        & $\mathbf{a}$   & {\bf ?}          & $\mathbf{a}$   & ?   & ...
\end{tabular}
\end{center}

We will analyze the sequence of rounds where signal $s_1$ is sent and the sequence of rounds where signal $s_2$ is sent, respectively.  Denote the two sequences by $Seq_{s_1}$ and $Seq_{s_2}$.  We will prove the following claims: 
\begin{claim}
As $t \to \infty$, 
\begin{enumerate}
    \item The length $L_{s_1}$ of the sequence $Seq_{s_1}$ is roughly $\frac{t}{2}$.
    \item The sender's total expected utility in $Seq_{s_1}$ is roughly $\frac{3}{4}L_{s_1}$.  
    \item The length $L_{s_2}$ of the sequence $Seq_{s_2}$ is roughly $\frac{t}{2}$.
    \item The sender's total expected utility in $Seq_{s_2}$ is roughly $\frac{1}{2}L_{s_2}$. 
\end{enumerate}
\end{claim}
The above claims imply that the sender's expected average utility in all the $t$ rounds is roughly: 
\begin{equation}
    \frac{1}{t}\big( \frac{3}{4}L_{s_1} + \frac{1}{2}L_{s_2}\big) \approx \frac{1}{t}\big( \frac{3}{8}t + \frac{1}{4}t \big) =  \frac{8}{5}.  
\end{equation}
This proves the example.

It remains to prove the claims. 
We first consider the sequence $Seq_{s_1}$ and prove Claim (2).  The states in this sequence are alternating: $GBGBGB...$.  When the state is $G$, the numbers of $G$ states and $B$ states in previous rounds are equal, namely, the empirical probability $\hat \mu_{s_1}^{(t-1)}(G) = \frac{1}{2}$, so the receiver will choose an action randomly, giving the sender an expected utility of $\frac{1}{2}$.   When the state is $B$, the number of $G$ states is greater than the number of $B$ states in previous rounds, so the receiver will take action $a$, giving the sender a utility of $1$.  Since half of the sequence is $G$ states and half is $B$ states, the total expected utility of the sender is $\frac{L_{s_1}}{2} \cdot \frac{1}{2} + \frac{L_{s_1}}{2}\cdot 1= \frac{3}{4}L_{s_1}$.

We then prove Claims (1) and (3), namely, the lengths $L_{s_1} \approx L_{s_2} \approx \frac{t}{2}$.  We note that the whole sequence of states and signals can be generated according to the Markov chain in Figure \ref{fig:markov-1}, where a node represents the state and the signal sent at that state, an edge represents the next state and has transition probability $1/2$, and the starting point is the bottom right $(B, s_2)$.  This Markov chain has a unique stationary distribution, where all the four nodes have an equal probability of $1/4$.  By the fundamental theorem of Markov chains, the long-term frequency of visiting nodes $(G, s_1)$ and $(B, s_1)$ is equal to the stationary probability of the two nodes, which is $\frac{1}{4} + \frac{1}{4} = \frac{1}{2}$.  So, as $t\to\infty$, half of the rounds have signal $s_1$ and the other half have $s_2$, namely, $L_{s_1} \approx L_{s_2} \approx \frac{t}{2}$. 

\begin{figure}[h]
\centering
    \begin{subfigure}[b]{0.48\textwidth}
         \centering
         \includegraphics[width=0.73\textwidth]{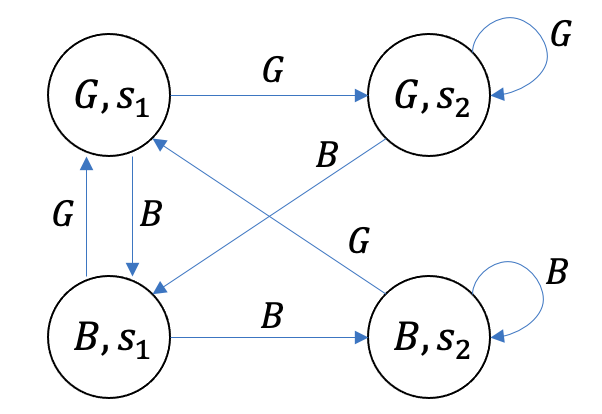}
         \caption{}
         \label{fig:markov-1}
    \end{subfigure}
\hfill  
    \begin{subfigure}[b]{0.48\textwidth}
         \centering
         \includegraphics[width=0.3\textwidth]{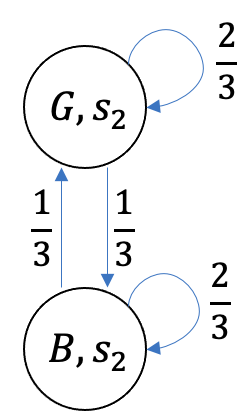}
         \caption{}
         \label{fig:markov-2}
    \end{subfigure}
\caption{Markov chains}
\end{figure}

Finally, we prove Claim (4) by analyzing the sequence $Seq_{s_2}$.  We only give the high-level idea.  We note that the sequence $Seq_{s_2}$ is generated by the Markov chain in Figure \ref{fig:markov-2}, which is the marginal Markov chain on $(G, s_2), (B, s_2)$ obtained from Figure \ref{fig:markov-1}.  We claim that the expected utility of the sender in each round of $Seq_{s_2}$ is roughly $\frac{1}{2}$.  Consider a round $t$ in $Seq_{s_2}$.  According to the Markov chain, we can see that the expected numbers of $G$ states and $B$ states in previous rounds are the same.  However, by a standard anti-concentration argument, with high probability the actual numbers of $G$ states and $B$ states in previous rounds must differ by at least $\Omega(\sqrt t)$: there are either more $G$s (in which case the receiver takes action $a$ at round $t$) or more $B$s (in which case the receiver takes $b$).  In order for the numbers of $G$s and $B$s to differ by $\Omega(\sqrt t)$ at round $t$, the two numbers must have already differed by at least $\frac{\Omega(t)}{2}$ at round $t - \frac{\Omega(\sqrt t)}{2}$. 
 In other words, whether the receiver will take action $a$ or $b$ at round $t$ has already been decided at round $t - \frac{\Omega(\sqrt t)}{2}$.  Say the receiver will take $a$.   Assuming that $t$ is large enough such that $\frac{\Omega(\sqrt t)}{2}$ is larger than the mixing time of the Markov chain, the Markov chain mixes during the time from round $t - \frac{\Omega(\sqrt t)}{2}$ to round $t$, so the distribution of state at round $t$ is roughly uniform: $\Pr[G] \approx \Pr[B] \approx \frac{1}{2}$.  Given that the receiver takes action $a$ at round $t$, the sender obtains expected utility $\Pr[G]\su(a, G) + \Pr[B]\su(a, B) \approx \frac{1}{2}$.
 Therefore, the sender's total expected utility in $Seq_{s_2}$ is $\frac{1}{2}L_{s_2}$. 
\end{proof}

We remark that in Example~\ref{ex:better-than-classic-model} the sender uses a time-varying signaling scheme.  If we impose the requirement that the sender uses a fixed signaling scheme for all rounds, then a similar result to Theorem~\ref{thm:BP-robustness} (or Lemma~\ref{lem:upper_bound}) can be proved: the sender cannot do much better than the classic model.  The proof is a combination of the proof of Lemma~\ref{lem:upper_bound} and the concentration bound in Lemma~\ref{lem:concentration}.  We omit the details here.

\section{Conclusion and Discussion}
\label{sec:discussion}
We have shown in this paper that, in the Bayesian persuasion problem, under a certain condition (Assumption~\ref{assump:unique_optimal_action_strong}), the sender's maximal achievable utility is not affected a lot by the receiver's approximately-best-responding behavior; the sender can achieve a similar utility as in the classic best-responding model.
The key intuition behind this result is that any direct-revelation scheme can be robustified under Assumption~\ref{assump:unique_optimal_action_strong}.
This result does not carry over to the learning agent setting, where the sender can sometimes achieve much higher utility than in the classic model. 
Assumption~\ref{assump:unique_optimal_action_strong} can be significantly relaxed as shown in our later paper \cite{lin2024persuading}. 

\bibliographystyle{ACM-Reference-Format}
\bibliography{bibs}


\begin{thebibliography}{41}


\ifx \showCODEN    \undefined \def \showCODEN     #1{\unskip}     \fi
\ifx \showDOI      \undefined \def \showDOI       #1{#1}\fi
\ifx \showISBNx    \undefined \def \showISBNx     #1{\unskip}     \fi
\ifx \showISBNxiii \undefined \def \showISBNxiii  #1{\unskip}     \fi
\ifx \showISSN     \undefined \def \showISSN      #1{\unskip}     \fi
\ifx \showLCCN     \undefined \def \showLCCN      #1{\unskip}     \fi
\ifx \shownote     \undefined \def \shownote      #1{#1}          \fi
\ifx \showarticletitle \undefined \def \showarticletitle #1{#1}   \fi
\ifx \showURL      \undefined \def \showURL       {\relax}        \fi
\providecommand\bibfield[2]{#2}
\providecommand\bibinfo[2]{#2}
\providecommand\natexlab[1]{#1}
\providecommand\showeprint[2][]{arXiv:#2}

\bibitem[Alonso and CÂmara(2016)]%
        {alonso_persuading_2016}
\bibfield{author}{\bibinfo{person}{Ricardo Alonso} {and}
  \bibinfo{person}{Odilon CÂmara}.} \bibinfo{year}{2016}\natexlab{}.
\newblock \showarticletitle{Persuading {Voters}}.
\newblock \bibinfo{journal}{\emph{American Economic Review}}
  \bibinfo{volume}{106}, \bibinfo{number}{11} (\bibinfo{date}{Nov.}
  \bibinfo{year}{2016}), \bibinfo{pages}{3590--3605}.
\newblock
\showISSN{0002-8282}
\urldef\tempurl%
\url{https://doi.org/10.1257/aer.20140737}
\showDOI{\tempurl}


\bibitem[Anunrojwong et~al\mbox{.}(2023)]%
        {anunrojwong_persuading_2023}
\bibfield{author}{\bibinfo{person}{Jerry Anunrojwong},
  \bibinfo{person}{Krishnamurthy Iyer}, {and} \bibinfo{person}{David
  Lingenbrink}.} \bibinfo{year}{2023}\natexlab{}.
\newblock \showarticletitle{Persuading {Risk}-{Conscious} {Agents}: {A}
  {Geometric} {Approach}}.
\newblock \bibinfo{journal}{\emph{Operations Research}} (\bibinfo{date}{March}
  \bibinfo{year}{2023}), \bibinfo{pages}{opre.2023.2438}.
\newblock
\showISSN{0030-364X, 1526-5463}
\urldef\tempurl%
\url{https://doi.org/10.1287/opre.2023.2438}
\showDOI{\tempurl}


\bibitem[Babichenko et~al\mbox{.}(2021)]%
        {babichenko_regret-minimizing_2021}
\bibfield{author}{\bibinfo{person}{Yakov Babichenko}, \bibinfo{person}{Inbal
  Talgam-Cohen}, \bibinfo{person}{Haifeng Xu}, {and}
  \bibinfo{person}{Konstantin Zabarnyi}.} \bibinfo{year}{2021}\natexlab{}.
\newblock \showarticletitle{Regret-{Minimizing} {Bayesian} {Persuasion}}. In
  \bibinfo{booktitle}{\emph{Proceedings of the 22nd {ACM} {Conference} on
  {Economics} and {Computation}}}. \bibinfo{publisher}{ACM},
  \bibinfo{address}{Budapest Hungary}, \bibinfo{pages}{128--128}.
\newblock
\showISBNx{978-1-4503-8554-1}
\urldef\tempurl%
\url{https://doi.org/10.1145/3465456.3467574}
\showDOI{\tempurl}


\bibitem[Benjamin(2019)]%
        {benjamin2019errors}
\bibfield{author}{\bibinfo{person}{Daniel~J Benjamin}.}
  \bibinfo{year}{2019}\natexlab{}.
\newblock \showarticletitle{Errors in probabilistic reasoning and judgment
  biases}.
\newblock \bibinfo{journal}{\emph{Handbook of Behavioral Economics:
  Applications and Foundations 1}}  \bibinfo{volume}{2} (\bibinfo{year}{2019}),
  \bibinfo{pages}{69--186}.
\newblock


\bibitem[Bergemann and Morris(2019)]%
        {bergemann_information_2019}
\bibfield{author}{\bibinfo{person}{Dirk Bergemann} {and}
  \bibinfo{person}{Stephen Morris}.} \bibinfo{year}{2019}\natexlab{}.
\newblock \showarticletitle{Information {Design}: {A} {Unified} {Perspective}}.
\newblock \bibinfo{journal}{\emph{Journal of Economic Literature}}
  \bibinfo{volume}{57}, \bibinfo{number}{1} (\bibinfo{date}{March}
  \bibinfo{year}{2019}), \bibinfo{pages}{44--95}.
\newblock
\showISSN{0022-0515}
\urldef\tempurl%
\url{https://doi.org/10.1257/jel.20181489}
\showDOI{\tempurl}


\bibitem[Braverman et~al\mbox{.}(2018)]%
        {braverman_selling_2018}
\bibfield{author}{\bibinfo{person}{Mark Braverman}, \bibinfo{person}{Jieming
  Mao}, \bibinfo{person}{Jon Schneider}, {and} \bibinfo{person}{Matt
  Weinberg}.} \bibinfo{year}{2018}\natexlab{}.
\newblock \showarticletitle{Selling to a {No}-{Regret} {Buyer}}. In
  \bibinfo{booktitle}{\emph{Proceedings of the 2018 {ACM} {Conference} on
  {Economics} and {Computation}}}. \bibinfo{publisher}{ACM},
  \bibinfo{address}{Ithaca NY USA}, \bibinfo{pages}{523--538}.
\newblock
\showISBNx{978-1-4503-5829-3}
\urldef\tempurl%
\url{https://doi.org/10.1145/3219166.3219233}
\showDOI{\tempurl}


\bibitem[Brocas and Carrillo(2007)]%
        {brocas_influence_2007}
\bibfield{author}{\bibinfo{person}{Isabelle Brocas} {and}
  \bibinfo{person}{Juan~D. Carrillo}.} \bibinfo{year}{2007}\natexlab{}.
\newblock \showarticletitle{Influence through {Ignorance}}.
\newblock \bibinfo{journal}{\emph{The RAND Journal of Economics}}
  \bibinfo{volume}{38}, \bibinfo{number}{4} (\bibinfo{year}{2007}),
  \bibinfo{pages}{931--947}.
\newblock
\showISSN{07416261}
\urldef\tempurl%
\url{http://www.jstor.org/stable/25046346}
\showURL{%
\tempurl}
\newblock
\shownote{Publisher: [RAND Corporation, Wiley]}.


\bibitem[Cai et~al\mbox{.}(2023)]%
        {cai2023selling}
\bibfield{author}{\bibinfo{person}{Linda Cai}, \bibinfo{person}{S~Matthew
  Weinberg}, \bibinfo{person}{Evan Wildenhain}, {and} \bibinfo{person}{Shirley
  Zhang}.} \bibinfo{year}{2023}\natexlab{}.
\newblock \showarticletitle{Selling to Multiple No-Regret Buyers}.
\newblock \bibinfo{journal}{\emph{arXiv preprint arXiv:2307.04175}}
  (\bibinfo{year}{2023}).
\newblock


\bibitem[Camara et~al\mbox{.}(2020)]%
        {camara_mechanisms_2020}
\bibfield{author}{\bibinfo{person}{Modibo~K. Camara}, \bibinfo{person}{Jason~D.
  Hartline}, {and} \bibinfo{person}{Aleck Johnsen}.}
  \bibinfo{year}{2020}\natexlab{}.
\newblock \showarticletitle{Mechanisms for a {No}-{Regret} {Agent}: {Beyond}
  the {Common} {Prior}}. In \bibinfo{booktitle}{\emph{2020 {IEEE} 61st {Annual}
  {Symposium} on {Foundations} of {Computer} {Science} ({FOCS})}}.
  \bibinfo{publisher}{IEEE}, \bibinfo{address}{Durham, NC, USA},
  \bibinfo{pages}{259--270}.
\newblock
\showISBNx{978-1-72819-621-3}
\urldef\tempurl%
\url{https://doi.org/10.1109/FOCS46700.2020.00033}
\showDOI{\tempurl}


\bibitem[Camerer(1998)]%
        {camerer1998bounded}
\bibfield{author}{\bibinfo{person}{Colin Camerer}.}
  \bibinfo{year}{1998}\natexlab{}.
\newblock \showarticletitle{Bounded rationality in individual decision making}.
\newblock \bibinfo{journal}{\emph{Experimental economics}} \bibinfo{volume}{1},
  \bibinfo{number}{2} (\bibinfo{year}{1998}), \bibinfo{pages}{163--183}.
\newblock


\bibitem[Camerer et~al\mbox{.}(2004)]%
        {10.1162/0033553041502225}
\bibfield{author}{\bibinfo{person}{Colin~F. Camerer}, \bibinfo{person}{Teck-Hua
  Ho}, {and} \bibinfo{person}{Juin-Kuan Chong}.}
  \bibinfo{year}{2004}\natexlab{}.
\newblock \showarticletitle{{A Cognitive Hierarchy Model of Games*}}.
\newblock \bibinfo{journal}{\emph{The Quarterly Journal of Economics}}
  \bibinfo{volume}{119}, \bibinfo{number}{3} (\bibinfo{date}{08}
  \bibinfo{year}{2004}), \bibinfo{pages}{861--898}.
\newblock


\bibitem[Castiglioni et~al\mbox{.}(2020)]%
        {castiglioni_online_2020}
\bibfield{author}{\bibinfo{person}{Matteo Castiglioni}, \bibinfo{person}{Andrea
  Celli}, \bibinfo{person}{Alberto Marchesi}, {and} \bibinfo{person}{Nicola
  Gatti}.} \bibinfo{year}{2020}\natexlab{}.
\newblock \showarticletitle{Online {Bayesian} {Persuasion}}. In
  \bibinfo{booktitle}{\emph{Advances in {Neural} {Information} {Processing}
  {Systems}}}, Vol.~\bibinfo{volume}{33}. \bibinfo{publisher}{Curran
  Associates, Inc.}, \bibinfo{pages}{16188--16198}.
\newblock
\urldef\tempurl%
\url{https://proceedings.neurips.cc/paper/2020/file/ba5451d3c91a0f982f103cdbe249bc78-Paper.pdf}
\showURL{%
\tempurl}


\bibitem[Castiglioni et~al\mbox{.}(2021)]%
        {castiglioni2021multi}
\bibfield{author}{\bibinfo{person}{Matteo Castiglioni},
  \bibinfo{person}{Alberto Marchesi}, \bibinfo{person}{Andrea Celli}, {and}
  \bibinfo{person}{Nicola Gatti}.} \bibinfo{year}{2021}\natexlab{}.
\newblock \showarticletitle{Multi-receiver online bayesian persuasion}. In
  \bibinfo{booktitle}{\emph{International Conference on Machine Learning}}.
  PMLR, \bibinfo{pages}{1314--1323}.
\newblock


\bibitem[Cesa-Bianchi and Lugosi(2006)]%
        {cesa-bianchi_prediction_2006}
\bibfield{author}{\bibinfo{person}{Nicolo Cesa-Bianchi} {and}
  \bibinfo{person}{Gabor Lugosi}.} \bibinfo{year}{2006}\natexlab{}.
\newblock \bibinfo{booktitle}{\emph{Prediction, {Learning}, and {Games}}}.
\newblock \bibinfo{publisher}{Cambridge University Press},
  \bibinfo{address}{Cambridge}.
\newblock
\showISBNx{978-0-511-54692-1}
\urldef\tempurl%
\url{https://doi.org/10.1017/CBO9780511546921}
\showDOI{\tempurl}


\bibitem[de~Clippel and Zhang(2022)]%
        {de_clippel_non-bayesian_2022}
\bibfield{author}{\bibinfo{person}{Geoffroy de Clippel} {and}
  \bibinfo{person}{Xu Zhang}.} \bibinfo{year}{2022}\natexlab{}.
\newblock \showarticletitle{Non-{Bayesian} {Persuasion}}.
\newblock \bibinfo{journal}{\emph{Journal of Political Economy}}
  \bibinfo{volume}{130}, \bibinfo{number}{10} (\bibinfo{date}{Oct.}
  \bibinfo{year}{2022}), \bibinfo{pages}{2594--2642}.
\newblock
\showISSN{0022-3808, 1537-534X}
\urldef\tempurl%
\url{https://doi.org/10.1086/720464}
\showDOI{\tempurl}


\bibitem[Deng et~al\mbox{.}(2019)]%
        {deng_strategizing_2019}
\bibfield{author}{\bibinfo{person}{Yuan Deng}, \bibinfo{person}{Jon Schneider},
  {and} \bibinfo{person}{Balasubramanian Sivan}.}
  \bibinfo{year}{2019}\natexlab{}.
\newblock \showarticletitle{Strategizing against {No}-regret {Learners}}. In
  \bibinfo{booktitle}{\emph{Advances in {Neural} {Information} {Processing}
  {Systems}}}, \bibfield{editor}{\bibinfo{person}{H.~Wallach},
  \bibinfo{person}{H.~Larochelle}, \bibinfo{person}{A.~Beygelzimer},
  \bibinfo{person}{F.~d' Alché-Buc}, \bibinfo{person}{E.~Fox}, {and}
  \bibinfo{person}{R.~Garnett}} (Eds.), Vol.~\bibinfo{volume}{32}.
  \bibinfo{publisher}{Curran Associates, Inc.}
\newblock
\urldef\tempurl%
\url{https://proceedings.neurips.cc/paper/2019/file/8b6dd7db9af49e67306feb59a8bdc52c-Paper.pdf}
\showURL{%
\tempurl}


\bibitem[Dughmi(2017)]%
        {dughmi_algorithmic_2017}
\bibfield{author}{\bibinfo{person}{Shaddin Dughmi}.}
  \bibinfo{year}{2017}\natexlab{}.
\newblock \showarticletitle{Algorithmic information structure design: a
  survey}.
\newblock \bibinfo{journal}{\emph{ACM SIGecom Exchanges}} \bibinfo{volume}{15},
  \bibinfo{number}{2} (\bibinfo{date}{Feb.} \bibinfo{year}{2017}),
  \bibinfo{pages}{2--24}.
\newblock
\showISSN{1551-9031}
\urldef\tempurl%
\url{https://doi.org/10.1145/3055589.3055591}
\showDOI{\tempurl}


\bibitem[Dughmi and Xu(2016)]%
        {dughmi_algorithmic_2016}
\bibfield{author}{\bibinfo{person}{Shaddin Dughmi} {and}
  \bibinfo{person}{Haifeng Xu}.} \bibinfo{year}{2016}\natexlab{}.
\newblock \showarticletitle{Algorithmic {Bayesian} persuasion}. In
  \bibinfo{booktitle}{\emph{Proceedings of the forty-eighth annual {ACM}
  symposium on {Theory} of {Computing}}}. \bibinfo{publisher}{ACM},
  \bibinfo{address}{Cambridge MA USA}, \bibinfo{pages}{412--425}.
\newblock
\showISBNx{978-1-4503-4132-5}
\urldef\tempurl%
\url{https://doi.org/10.1145/2897518.2897583}
\showDOI{\tempurl}


\bibitem[Dworczak and Pavan(2020)]%
        {dworczak2020preparing}
\bibfield{author}{\bibinfo{person}{Piotr Dworczak} {and}
  \bibinfo{person}{Alessandro Pavan}.} \bibinfo{year}{2020}\natexlab{}.
\newblock \showarticletitle{Preparing for the worst but hoping for the best:
  Robust (bayesian) persuasion}.
\newblock  (\bibinfo{year}{2020}).
\newblock


\bibitem[Emek et~al\mbox{.}(2014)]%
        {emek_signaling_2014}
\bibfield{author}{\bibinfo{person}{Yuval Emek}, \bibinfo{person}{Michal
  Feldman}, \bibinfo{person}{Iftah Gamzu}, \bibinfo{person}{Renato PaesLeme},
  {and} \bibinfo{person}{Moshe Tennenholtz}.} \bibinfo{year}{2014}\natexlab{}.
\newblock \showarticletitle{Signaling {Schemes} for {Revenue} {Maximization}}.
\newblock \bibinfo{journal}{\emph{ACM Transactions on Economics and
  Computation}} \bibinfo{volume}{2}, \bibinfo{number}{2} (\bibinfo{date}{June}
  \bibinfo{year}{2014}), \bibinfo{pages}{1--19}.
\newblock
\showISSN{2167-8375, 2167-8383}
\urldef\tempurl%
\url{https://doi.org/10.1145/2594564}
\showDOI{\tempurl}


\bibitem[Feng et~al\mbox{.}(2024)]%
        {feng2024rationality}
\bibfield{author}{\bibinfo{person}{Yiding Feng}, \bibinfo{person}{Chien-Ju Ho},
  {and} \bibinfo{person}{Wei Tang}.} \bibinfo{year}{2024}\natexlab{}.
\newblock \showarticletitle{Rationality-robust information design: Bayesian
  persuasion under quantal response}. In \bibinfo{booktitle}{\emph{Proceedings
  of the 2024 Annual ACM-SIAM Symposium on Discrete Algorithms (SODA)}}. SIAM,
  \bibinfo{pages}{501--546}.
\newblock


\bibitem[Feng et~al\mbox{.}(2022)]%
        {feng_online_2022}
\bibfield{author}{\bibinfo{person}{Yiding Feng}, \bibinfo{person}{Wei Tang},
  {and} \bibinfo{person}{Haifeng Xu}.} \bibinfo{year}{2022}\natexlab{}.
\newblock \showarticletitle{Online {Bayesian} {Recommendation} with {No}
  {Regret}}. In \bibinfo{booktitle}{\emph{Proceedings of the 23rd {ACM}
  {Conference} on {Economics} and {Computation}}}. \bibinfo{publisher}{ACM},
  \bibinfo{address}{Boulder CO USA}, \bibinfo{pages}{818--819}.
\newblock
\showISBNx{978-1-4503-9150-4}
\urldef\tempurl%
\url{https://doi.org/10.1145/3490486.3538327}
\showDOI{\tempurl}


\bibitem[Gentzkow and Kamenica(2017)]%
        {gentzkow_bayesian_2017}
\bibfield{author}{\bibinfo{person}{Matthew Gentzkow} {and}
  \bibinfo{person}{Emir Kamenica}.} \bibinfo{year}{2017}\natexlab{}.
\newblock \showarticletitle{Bayesian persuasion with multiple senders and rich
  signal spaces}.
\newblock \bibinfo{journal}{\emph{Games and Economic Behavior}}
  \bibinfo{volume}{104} (\bibinfo{date}{July} \bibinfo{year}{2017}),
  \bibinfo{pages}{411--429}.
\newblock
\showISSN{08998256}
\urldef\tempurl%
\url{https://doi.org/10.1016/j.geb.2017.05.004}
\showDOI{\tempurl}


\bibitem[Guruganesh et~al\mbox{.}(2024)]%
        {guruganesh2024contracting}
\bibfield{author}{\bibinfo{person}{Guru Guruganesh}, \bibinfo{person}{Yoav
  Kolumbus}, \bibinfo{person}{Jon Schneider}, \bibinfo{person}{Inbal
  Talgam-Cohen}, \bibinfo{person}{Emmanouil-Vasileios Vlatakis-Gkaragkounis},
  \bibinfo{person}{Joshua~R. Wang}, {and} \bibinfo{person}{S.~Matthew
  Weinberg}.} \bibinfo{year}{2024}\natexlab{}.
\newblock \bibinfo{title}{Contracting with a Learning Agent}.
\newblock
\newblock
\showeprint[arxiv]{2401.16198}~[cs.GT]


\bibitem[Kamenica(2019)]%
        {kamenica2019bayesian}
\bibfield{author}{\bibinfo{person}{Emir Kamenica}.}
  \bibinfo{year}{2019}\natexlab{}.
\newblock \showarticletitle{Bayesian persuasion and information design}.
\newblock \bibinfo{journal}{\emph{Annual Review of Economics}}
  \bibinfo{volume}{11} (\bibinfo{year}{2019}), \bibinfo{pages}{249--272}.
\newblock


\bibitem[Kamenica and Gentzkow(2011)]%
        {kamenica_bayesian_2011}
\bibfield{author}{\bibinfo{person}{Emir Kamenica} {and}
  \bibinfo{person}{Matthew Gentzkow}.} \bibinfo{year}{2011}\natexlab{}.
\newblock \showarticletitle{Bayesian {Persuasion}}.
\newblock \bibinfo{journal}{\emph{American Economic Review}}
  \bibinfo{volume}{101}, \bibinfo{number}{6} (\bibinfo{date}{Oct.}
  \bibinfo{year}{2011}), \bibinfo{pages}{2590--2615}.
\newblock
\showISSN{0002-8282}
\urldef\tempurl%
\url{https://doi.org/10.1257/aer.101.6.2590}
\showDOI{\tempurl}


\bibitem[Kosterina(2022)]%
        {kosterina_persuasion_2022}
\bibfield{author}{\bibinfo{person}{Svetlana Kosterina}.}
  \bibinfo{year}{2022}\natexlab{}.
\newblock \showarticletitle{Persuasion with unknown beliefs}.
\newblock \bibinfo{journal}{\emph{Theoretical Economics}} \bibinfo{volume}{17},
  \bibinfo{number}{3} (\bibinfo{year}{2022}), \bibinfo{pages}{1075--1107}.
\newblock
\showISSN{1933-6837}
\urldef\tempurl%
\url{https://doi.org/10.3982/TE4742}
\showDOI{\tempurl}


\bibitem[Lin and Chen(2024)]%
        {lin2024persuading}
\bibfield{author}{\bibinfo{person}{Tao Lin} {and} \bibinfo{person}{Yiling
  Chen}.} \bibinfo{year}{2024}\natexlab{}.
\newblock \bibinfo{title}{Persuading a Learning Agent}.
\newblock
\newblock
\showeprint[arxiv]{2402.09721}~[cs.GT]


\bibitem[Mansour et~al\mbox{.}(2022)]%
        {mansour_strategizing_2022}
\bibfield{author}{\bibinfo{person}{Yishay Mansour}, \bibinfo{person}{Mehryar
  Mohri}, \bibinfo{person}{Jon Schneider}, {and}
  \bibinfo{person}{Balasubramanian Sivan}.} \bibinfo{year}{2022}\natexlab{}.
\newblock \showarticletitle{Strategizing against {Learners} in {Bayesian}
  {Games}}. In \bibinfo{booktitle}{\emph{Proceedings of {Thirty} {Fifth}
  {Conference} on {Learning} {Theory}}} \emph{(\bibinfo{series}{Proceedings of
  {Machine} {Learning} {Research}}, Vol.~\bibinfo{volume}{178})}.
  \bibinfo{publisher}{PMLR}, \bibinfo{pages}{5221--5252}.
\newblock
\urldef\tempurl%
\url{https://proceedings.mlr.press/v178/mansour22a.html}
\showURL{%
\tempurl}


\bibitem[McKelvey and Palfrey(1995)]%
        {mckelvey_quantal_1995}
\bibfield{author}{\bibinfo{person}{Richard~D. McKelvey} {and}
  \bibinfo{person}{Thomas~R. Palfrey}.} \bibinfo{year}{1995}\natexlab{}.
\newblock \showarticletitle{Quantal {Response} {Equilibria} for {Normal} {Form}
  {Games}}.
\newblock \bibinfo{journal}{\emph{Games and Economic Behavior}}
  \bibinfo{volume}{10}, \bibinfo{number}{1} (\bibinfo{date}{July}
  \bibinfo{year}{1995}), \bibinfo{pages}{6--38}.
\newblock
\showISSN{08998256}
\urldef\tempurl%
\url{https://doi.org/10.1006/game.1995.1023}
\showDOI{\tempurl}


\bibitem[Rabin(2013)]%
        {10.1257/aer.103.3.617}
\bibfield{author}{\bibinfo{person}{Matthew Rabin}.}
  \bibinfo{year}{2013}\natexlab{}.
\newblock \showarticletitle{An Approach to Incorporating Psychology into
  Economics}.
\newblock \bibinfo{journal}{\emph{American Economic Review}}
  \bibinfo{volume}{103}, \bibinfo{number}{3} (\bibinfo{date}{May}
  \bibinfo{year}{2013}), \bibinfo{pages}{617--22}.
\newblock
\urldef\tempurl%
\url{https://doi.org/10.1257/aer.103.3.617}
\showDOI{\tempurl}


\bibitem[Ravindran and Cui(2022)]%
        {ravindran_competing_2022}
\bibfield{author}{\bibinfo{person}{Dilip Ravindran} {and}
  \bibinfo{person}{Zhihan Cui}.} \bibinfo{year}{2022}\natexlab{}.
\newblock \bibinfo{title}{Competing {Persuaders} in {Zero}-{Sum} {Games}}.
\newblock
\newblock
\urldef\tempurl%
\url{http://arxiv.org/abs/2008.08517}
\showURL{%
\tempurl}
\newblock
\shownote{arXiv:2008.08517 [econ]}.


\bibitem[Rubinstein and Zhao(2024)]%
        {rubinstein2024strategizing}
\bibfield{author}{\bibinfo{person}{Aviad Rubinstein} {and}
  \bibinfo{person}{Junyao Zhao}.} \bibinfo{year}{2024}\natexlab{}.
\newblock \bibinfo{title}{Strategizing against No-Regret Learners in
  First-Price Auctions}.
\newblock
\newblock
\showeprint[arxiv]{2402.08637}~[cs.GT]


\bibitem[Stahl and Wilson(1995)]%
        {Stahl1995OnPM}
\bibfield{author}{\bibinfo{person}{Dale~O. Stahl} {and}
  \bibinfo{person}{Paul~W. Wilson}.} \bibinfo{year}{1995}\natexlab{}.
\newblock \showarticletitle{On Players' Models of Other Players: Theory and
  Experimental Evidence}.
\newblock \bibinfo{journal}{\emph{Games and Economic Behavior}}
  \bibinfo{volume}{10} (\bibinfo{year}{1995}), \bibinfo{pages}{218--254}.
\newblock


\bibitem[Taneva(2019)]%
        {taneva_information_2019}
\bibfield{author}{\bibinfo{person}{Ina Taneva}.}
  \bibinfo{year}{2019}\natexlab{}.
\newblock \showarticletitle{Information {Design}}.
\newblock \bibinfo{journal}{\emph{American Economic Journal: Microeconomics}}
  \bibinfo{volume}{11}, \bibinfo{number}{4} (\bibinfo{date}{Nov.}
  \bibinfo{year}{2019}), \bibinfo{pages}{151--185}.
\newblock
\showISSN{1945-7669, 1945-7685}
\urldef\tempurl%
\url{https://doi.org/10.1257/mic.20170351}
\showDOI{\tempurl}


\bibitem[{Tyler Lu} et~al\mbox{.}(2010)]%
        {tyler_lu_contextual_2010}
\bibfield{author}{\bibinfo{person}{{Tyler Lu}}, \bibinfo{person}{{David Pal}},
  {and} \bibinfo{person}{{Martin Pal}}.} \bibinfo{year}{2010}\natexlab{}.
\newblock \showarticletitle{Contextual {Multi}-{Armed} {Bandits}}. In
  \bibinfo{booktitle}{\emph{Proceedings of the {Thirteenth} {International}
  {Conference} on {Artificial} {Intelligence} and {Statistics}}},
  Vol.~\bibinfo{volume}{9}. \bibinfo{publisher}{PMLR},
  \bibinfo{pages}{485--492}.
\newblock
\urldef\tempurl%
\url{https://proceedings.mlr.press/v9/lu10a.html}
\showURL{%
\tempurl}


\bibitem[Wang(2013)]%
        {wang_bayesian_2013}
\bibfield{author}{\bibinfo{person}{Yun Wang}.} \bibinfo{year}{2013}\natexlab{}.
\newblock \showarticletitle{Bayesian {Persuasion} with {Multiple} {Receivers}}.
\newblock \bibinfo{journal}{\emph{SSRN Electronic Journal}}
  (\bibinfo{year}{2013}).
\newblock
\showISSN{1556-5068}
\urldef\tempurl%
\url{https://doi.org/10.2139/ssrn.2625399}
\showDOI{\tempurl}


\bibitem[Wu et~al\mbox{.}(2022)]%
        {wu_sequential_2022}
\bibfield{author}{\bibinfo{person}{Jibang Wu}, \bibinfo{person}{Zixuan Zhang},
  \bibinfo{person}{Zhe Feng}, \bibinfo{person}{Zhaoran Wang},
  \bibinfo{person}{Zhuoran Yang}, \bibinfo{person}{Michael~I. Jordan}, {and}
  \bibinfo{person}{Haifeng Xu}.} \bibinfo{year}{2022}\natexlab{}.
\newblock \showarticletitle{Sequential {Information} {Design}: {Markov}
  {Persuasion} {Process} and {Its} {Efficient} {Reinforcement} {Learning}}. In
  \bibinfo{booktitle}{\emph{Proceedings of the 23rd {ACM} {Conference} on
  {Economics} and {Computation}}}. \bibinfo{publisher}{ACM},
  \bibinfo{address}{Boulder CO USA}, \bibinfo{pages}{471--472}.
\newblock
\showISBNx{978-1-4503-9150-4}
\urldef\tempurl%
\url{https://doi.org/10.1145/3490486.3538313}
\showDOI{\tempurl}


\bibitem[Yang and Zhang(2024)]%
        {yang2024computational}
\bibfield{author}{\bibinfo{person}{Kunhe Yang} {and} \bibinfo{person}{Hanrui
  Zhang}.} \bibinfo{year}{2024}\natexlab{}.
\newblock \bibinfo{title}{Computational Aspects of Bayesian Persuasion under
  Approximate Best Response}.
\newblock
\newblock
\showeprint[arxiv]{2402.07426}~[cs.GT]


\bibitem[Ziegler(2020)]%
        {ziegler2020adversarial}
\bibfield{author}{\bibinfo{person}{Gabriel Ziegler}.}
  \bibinfo{year}{2020}\natexlab{}.
\newblock \bibinfo{booktitle}{\emph{Adversarial bilateral information design}}.
\newblock \bibinfo{type}{{T}echnical {R}eport}.
\newblock


\bibitem[Zu et~al\mbox{.}(2021)]%
        {zu_learning_2021}
\bibfield{author}{\bibinfo{person}{You Zu}, \bibinfo{person}{Krishnamurthy
  Iyer}, {and} \bibinfo{person}{Haifeng Xu}.} \bibinfo{year}{2021}\natexlab{}.
\newblock \showarticletitle{Learning to {Persuade} on the {Fly}: {Robustness}
  {Against} {Ignorance}}. In \bibinfo{booktitle}{\emph{Proceedings of the 22nd
  {ACM} {Conference} on {Economics} and {Computation}}}.
  \bibinfo{publisher}{ACM}, \bibinfo{address}{Budapest Hungary},
  \bibinfo{pages}{927--928}.
\newblock
\showISBNx{978-1-4503-8554-1}
\urldef\tempurl%
\url{https://doi.org/10.1145/3465456.3467593}
\showDOI{\tempurl}


\end{thebibliography}

\appendix
\section{Missing Proofs from Section~\ref{sec:model}}

\subsection{Proof of Example \ref{example:approximately-best-responding}}
\label{app:example:approximately-best-responding}
Consider the quantal response model.  Let $\gamma = \frac{\log(|A| \lambda)}{\lambda}$.  Given signal $s$, if an action $a\in A$ is \emph{not} a $\gamma$-best-responding action, then by definition
\[ \ru(a^*_\pi(s), \mu_s) - \ru(a, \mu_s)  \ge \gamma\]
where $a^*_\pi(s)$ is a best-responding action.  So, the probability of the receiver choosing $a$ is at most: 
\begin{align*}
    \frac{\exp(\lambda \ru(a, \mu_s))}{\sum_{a\in A} \exp(\lambda \ru(a, \mu_s))} \le \frac{\exp(\lambda \ru(a, \mu_s))}{\exp(\lambda \ru(a^*_\pi(s), \mu_s))} = \exp\big( - \lambda \big( \ru(a^*_\pi(s), \mu_s ) - \ru(a, \mu_s) \big) \big) \le \exp( - \lambda \gamma ) = \frac{1}{|A|\lambda}. 
\end{align*}
By a union bound, the probability of the receiver choosing any not $\gamma$-approxiamtely optimal action is at most $\frac{1}{\lambda}$.  So, this strategy is $(\frac{\log(|A| \lambda)}{\lambda}, \frac{1}{\lambda})$-best-responding.

\subsection{Proof of Lemma \ref{lem:restrict_to_gamma-best}}
\label{app:lem:restrict_to_gamma-best}
For any $(\gamma, \delta)$-best-responding receiver strategy $\rho$, we define receiver strategy $\tilde \rho$ that always takes $\gamma$-best-responding actions and takes those actions with probabilities proportional to the probabilities in $\rho$: formally, 
\begin{align*}
	\rho(a \smid s) =
	\begin{cases}
		\frac{\rho(a \smid s)}{\rho(A_{\pi}^\gamma(s) \smid s)}  & \text{if } a \in A_{\pi}^\gamma(s), \\
		0 & \text{otherwise,}  
	\end{cases} 
\end{align*}
where $\rho(A_{\pi}^\gamma(s) \smid s) = \sum_{a\in A_{\pi}^\gamma(s)} \rho(a \smid s) = \Pr_{a\sim \rho(s)}[a\in A_\pi^\gamma(s) ]\ge 1 - \delta > 0$ by the definition of $(\gamma, \delta)$-best-responding strategy.

We first prove $\esu_\prior(\pi, \rho) \le \esu_\prior(\pi, \tilde \rho) +  \delta$: 
\begin{align*}
	\esu_\prior(\pi, \rho) & = \sum_{(\omega, s)} \pi(\omega, s) \E_{a\sim \rho(s)}\big[ \su(a, \omega) \big] \\
	& = \sum_{(\omega, s)} \pi(\omega, s) \Big( \E_{a\sim \rho(s)}\big[ \su(a, \omega) \mid a\in A_\pi^\gamma(s) \big] \cdot \Pr_{a\sim \rho(s)}\big[a\in A_\pi^\gamma(s)\big] \\
	& \hspace{6em} + \E_{a\sim \rho(s)}\big[ \su(a, \omega) \mid a\notin A_\pi^\gamma(s) \big] \cdot \Pr_{a\sim \rho(s)}\big[a\notin A_\pi^\gamma(s)\big]  \Big) \\
	& \le \sum_{(\omega, s)} \pi(\omega, s) \Big( \E_{a\sim \rho(s)}\big[ \su(a, \omega) \mid a\in A_\pi^\gamma(s) \big] \cdot 1  \; + \; 1 \cdot \delta  \Big) \\
	& = \sum_{(\omega, s)} \pi(\omega, s) \Big( \E_{a\sim \tilde \rho(s)}\big[ \su(a, \omega) \big]  \; + \; \delta \Big) \\
	& =  \esu_\prior(\pi, \tilde \rho) +  \delta. 
\end{align*}

We then prove $\esu_\prior(\pi, \rho) \ge \esu_\prior(\pi, \tilde \rho) - \delta$: 
\begin{align*}
	\esu_\prior(\pi, \rho) & = \sum_{(\omega, s)} \pi(\omega, s) \E_{a\sim \rho(s)}\big[ \su(a, \omega) \big] \\
	& \ge \sum_{(\omega, s)} \pi(\omega, s) \Big( \E_{a\sim \rho(s)}\big[ \su(a, \omega) \mid a\in A_\pi^\gamma(s) \big] \cdot \Pr_{a\sim \rho(s)}\big[a\in A_\pi^\gamma(s)\big]  \; + \; 0 \Big)  \\
	& \ge \sum_{(\omega, s)} \pi(\omega, s) \E_{a\sim \rho(s)}\big[ \su(a, \omega) \mid a\in A_\pi^\gamma(s) \big] \cdot (1-\delta) \\
	& = (1-\delta) \sum_{(\omega, s)} \pi(\omega, s) \E_{a\sim \tilde \rho(s)}\big[ \su(a, \omega) \big] \\
	& = (1-\delta) \esu_\prior(\pi, \tilde \rho) \\
	& \ge \esu_\prior(\pi, \tilde \rho) -  \delta. 
\end{align*}

Thus, $\big|\esu_\prior(\pi, \rho) - \esu_\prior(\pi, \tilde \rho)\big| \le \delta$ is proved.

\section{Missing Proofs from Section \ref{sec:main-result}}
\subsection{Proof of Claim \ref{claim:implies}}
\label{app:proof_claims:implies}
By Lemma~\ref{lem:restrict_to_gamma-best}, for any signaling scheme $\pi$, any $(\gamma, \delta)$-best-responding receiver strategy $\rho$,  there exists a $\gamma$-best-responding receiver strategy $\tilde \rho$ such that $|\esu_\prior(\pi, \rho) - \esu_\prior(\pi, \tilde \rho)| \le \delta$.
So, for $\underline{\OBJ}(\prior, \gamma, \delta)$ we have  
\begin{align*}
\underline{\OBJ}(\prior, \gamma, \delta) = \sup_{\pi} \inf_{\rho: (\gamma, \delta)\text{-best-responding}} \esu_\prior(\pi, \rho) & \ge  \sup_{\pi} \inf_{\tilde \rho: \gamma\text{-best-responding}} \esu_\prior(\pi, \tilde \rho) - \delta \\
& = \underline{\OBJ}(\prior, \gamma, 0) - \delta \\
\text{by \eqref{eq:main-result-gamma-best-responding} } & \ge \OPT^\BP(\prior) - \tfrac{\gamma}{\prior_{\min} (\Delta-\gamma)} - \delta 
\end{align*}
and for $\overline{\OBJ}(\prior, \gamma, \delta)$ we have  
\begin{align*}
\overline{\OBJ}(\prior, \gamma, \delta) = \sup_{\pi} \sup_{\rho: (\gamma, \delta)\text{-best-responding}} \esu_\prior(\pi, \rho) & \le  \sup_{\pi} \sup_{\tilde \rho: \gamma\text{-best-responding}} \esu_\prior(\pi, \tilde \rho) + \delta \\
& = \overline{\OBJ}(\prior, \gamma, 0) + \delta \\
\text{by \eqref{eq:main-result-gamma-best-responding} } & \le \OPT^\BP(\prior) + \tfrac{\gamma}{\prior_{\min} \Delta} + \delta. 
\end{align*}

\subsection{Proof of Claim \ref{claim:deterministic-direct}}
\label{app:claim:deterministic-direct}
Under the definition \eqref{eq:pi-direct-definition} of $\pi^{\mathrm{direct}}$, the posterior probability of a state $\omega \in \Omega$ given signal $a \in A$ is: 
\begin{align*}
	\pi^{\mathrm{direct}}(\omega \smid a) & = \frac{\prior(\omega) \pi^{\mathrm{direct}}(a \smid \omega)}{\pi^{\mathrm{direct}}(a)} = \frac{\prior(\omega) \pi(S_a \smid \omega)}{\pi(S_a)} \\
	& = \frac{\prior(\omega) \sum_{s\in S_a} \pi(s \smid \omega)}{\sum_{s\in S_a} \pi(s)} = \frac{\sum_{s\in S_a} \pi(s) \cdot \frac{\prior(\omega) \pi(s \smid \omega)}{\pi(s)}}{\sum_{s\in S_a} \pi(s)} = \frac{\sum_{s\in S_a} \pi(s) \cdot \pi(\omega \smid s)}{\sum_{s\in S_a} \pi(s)} 
\end{align*}

To prove claim (1), we note that, when the receiver is recommended action $a\in A$, the difference between its expected utilities of taking action $a$ and any other action $a' \ne a$ is: 
\begin{align*}
	\sum_{\omega \in \Omega} \pi^{\mathrm{direct}}(\omega \smid a)  \Big( \ru(a, \omega) - \ru(a', \omega) \Big) & = 
	\sum_{\omega \in \Omega} \frac{\sum_{s\in S_a} \pi(s) \cdot \pi(\omega \smid s)}{\sum_{s\in S_a} \pi(s)}  \Big( \ru(a, \omega) - \ru(a', \omega) \Big) \\
	& = \frac{1}{\sum_{s\in S_a} \pi(s)} \sum_{s\in S_a} \pi(s) \sum_{\omega \in \Omega} \pi(\omega \smid s) \Big( \ru(a, \omega) - \ru(a', \omega) \Big).
\end{align*}
Because $\rho^{\det}$ is a $\gamma$-best-responding strategy and $\rho^{\det}(s) = a$ for $s\in S_a$, $a$ must be a $\gamma$-best-responding action to $s$ under signaling scheme $\pi$, namely, $\sum_{\omega \in \Omega} \pi(\omega \smid s) \big( \ru(a, \omega) - \ru(a', \omega) \big) \ge -\gamma$.  So, 
\begin{align*}
	\sum_{\omega \in \Omega} \pi^{\mathrm{direct}}(\omega \smid a)  \Big( \ru(a, \omega) - \ru(a', \omega) \Big)
	~ \ge~  \frac{1}{\sum_{s\in S_a} \pi(s)} \sum_{s\in S_a} \pi(s) \cdot (-\gamma) ~ =~  -\gamma. 
\end{align*}
This implies that the obedient strategy is $\gamma$-best-responding. 

To prove claim (2), we note that the sender's expected utility under $\pi^{\mathrm{direct}}$ when the receiver takes the obedient strategy is: 
\begin{align*}
	\esu_\prior(\pi^{\mathrm{direct}}, \rho^{\ob}) & = \sum_{a\in A} \pi^{\mathrm{direct}}(a) \sum_{\omega\in\Omega} \pi^{\mathrm{direct}}(\omega \smid a) \cdot \su(a, \omega) \\
	& =  \sum_{a\in A} \pi(S_a) \sum_{\omega\in\Omega} \frac{\sum_{s\in S_a} \pi(s) \cdot \pi(\omega \smid s)}{\sum_{s\in S_a} \pi(s)}  \cdot \su(a, \omega) \\
	& =  \sum_{a\in A} \sum_{\omega\in\Omega} \sum_{s\in S_a} \pi(s) \cdot \pi(\omega \smid s) \cdot \su(a, \omega) \\
	& =  \sum_{a\in A} \sum_{s\in S_a} \pi(s) \sum_{\omega\in\Omega} \pi(\omega \smid s) \cdot \su(a, \omega) \\
	& =  \sum_{s\in S} \pi(s) \sum_{\omega\in\Omega} \pi(\omega \smid s) \cdot \su(a = \rho^{\mathrm{det}}(s), \omega) \\
	& = \esu_\prior(\pi, \rho^{\mathrm{det}}). 
\end{align*}

\subsection{Missing Proofs from Section~\ref{sec:learning}}
\subsection{Proof of Example \ref{ex:infinite-empirical-best-responding}}
\label{app:ex:infinite-empirical-best-responding}
Let $T_s^{(t-1)}$ be the number of rounds in the first $t-1$ rounds where the signal is $s$, namely, $ T_s^{(t-1)} = \sum_{\tau \le t-1} \mathbbm{1}[s^{(\tau)} = s]$.  By definition, we have
\begin{align*}
    \ru(a, \hat \mu_s^{(t-1)}) = \frac{1}{T_s^{(t-1)}} \sum_{\tau \le t-1, s^{(\tau)} = s} \ru(a, \omega^{(\tau)}). 
\end{align*}
So, the probability with which the algorithm chooses action $a$ at round $t$ is equal to
\begin{align*}
    \frac{\exp\big(\eta_t \sum_{\tau \le t-1, s^{(\tau)} = s} \ru(a, \omega^{(\tau)}) \big)}{\sum_{a'\in A} \exp\big(\eta_t \sum_{\tau \le t-1, s^{(\tau)} = s} \ru(a, \omega^{(\tau)}) \big)} = \frac{\exp\big(\eta_t T_s^{(t-1)} \ru(a, \hat \mu_s^{(t-1)}) \big)}{\sum_{a'\in A} \exp\big(\eta_t T_s^{(t-1)} \ru(a', \hat \mu_s^{(t-1)}) \big)}. 
\end{align*}
We note that this probability is equal to the probability in the quantal response model (Example \ref{example:approximately-best-responding}) with the parameter $\lambda = \eta_t T_s^{(t-1)}$ and the posterior distribution $\mu_s$ replaced by $\hat \mu_s^{(t-1)}$.  So, according to the result for the quantal response model, the algorithm here is $(\gamma_t, \delta_t)$-best-responding to $\hat \mu_s^{(t-1)}$ with
\begin{align*}
    \gamma_t = \frac{\log(|A|\lambda)}{\lambda} = \frac{\log(|A|\eta_t T_s^{(t-1)})}{\eta_t T_s^{(t-1)}}, \quad \quad \delta_t = \frac{1}{\lambda} = \frac{1}{\eta_t T_s^{(t-1)}}. 
\end{align*}
By our choice of $\eta_t$, $\eta_t = O(\sqrt{\frac{\log|A|}{t}})$.   By the law of large number and by our assumption that every signal $s$ is sent with probability at least $p$ at each round, we have $T_s^{(t-1)} \ge (t-1) p \approx t p$ when $t$ is large.  Plugging in $\eta_t$ and $T_s^{(t-1)}$, we obtain
\begin{align*}
    \gamma_t \le O\Big( \frac{\log\big(|A| \sqrt{\frac{\log|A|}{t}} p t \big)}{ \sqrt{\frac{\log|A|}{t}} p t } \Big) \le O\Big( \frac{\log(|A|t)}{p \sqrt{t \log|A|}}  \Big), \quad \quad \delta_t \le O\Big( \frac{1}{\sqrt{\frac{\log|A|}{t}} p t} \Big) = O\Big(\frac{1}{p \sqrt{t\log|A|}} \Big). 
\end{align*}

\subsection{Proof of Lemma \ref{lem:concentration}}
\label{app:lem:concentration}
Let $\hat \pi^{(t)}(s)$ be the empirical marginal probability of signal $s$ in the first $t$ rounds: 
\[ \hat \pi^{(t)}(s) = \frac{1}{t} \sum_{\tau=1}^t \mathbbm{1}[s^{(\tau)} = s]. \]
Since the signaling scheme $\pi$ is fixed in all rounds, we have
\[ \E[\hat \pi^{(t)}(s)] = \frac{1}{t} \sum_{\tau=1}^t \E\big[ \mathbbm{1}[s^{(\tau)} = s] \big] = \frac{1}{t} \sum_{\tau=1}^t \pi(s) = \pi(s). \]
Consider the empirical quantity $\hat \pi^{(t)}(s) \ru(a, \hat \mu_s^{(t)})$: 
\begin{align} 
    \hat \pi^{(t)}(s) \ru(a, \hat \mu_s^{(t)}) & = \frac{1}{t} \sum_{\tau=1}^t \mathbbm{1}[s^{(\tau)} = s] \frac{\sum_{\tau=1}^t \mathbbm{1}[s^{(\tau)}=s] \ru(a, \omega^{(\tau)})}{\sum_{\tau=1}^t \mathbbm{1}[s^{(\tau)} = s]}  \nonumber \\
    & = \frac{1}{t} \sum_{\tau=1}^t \mathbbm{1}[s^{(\tau)}=s] \ru(a, \omega^{(\tau)}).  \label{eq:empirical-quantity}
\end{align}
We note that its expectation equals
\begin{align*}
    \E\big[ \hat \pi^{(t)}(s) \ru(a, \hat \mu_s^{(t)}) \big] & = \frac{1}{t} \sum_{\tau=1}^t \E\big[ \mathbbm{1}[s^{(\tau)}=s] \ru(a, \omega^{(\tau)}) \big] \\
    & = \frac{1}{t} \sum_{\tau=1}^t \pi(s) \E_{\omega | s}\big[\ru(a, \omega) \smid s\big]  = \pi(s) \ru(a, \mu_s). 
\end{align*}
According to \eqref{eq:empirical-quantity}, $\hat \pi^{(t)}(s) \ru(a, \hat \mu_s^{(t)})$ is the average of $t$ i.i.d.~random variables in $[0, 1]$.  So, by Hoeffding's inequality, for any $\eps \ge 0$, 
\begin{align*}
    \Pr\Big[ \big|\hat \pi^{(t)}(s) \ru(a, \hat \mu_s^{(t)}) - \pi(s) \ru(a, \mu_s) \big| > \eps \Big] \le 2e^{-2t\eps^2}. 
\end{align*}
Using a union bound over $s\in S$ and $a\in A$, and letting $\eps = \sqrt{\frac{\log(2|S||A|t)}{2t}}$, we obtain
\begin{align*}
    \Pr\Big[ \exists s\in S, a\in A, \text{ s.t. } \big|\hat \pi^{(t)}(s) \ru(a, \hat \mu_s^{(t)}) - \pi(s) \ru(a, \mu_s) \big| > \eps \Big] \le 2|S||A|e^{-2t\eps^2} = \frac{1}{t}. 
\end{align*}
In other words, with probability at least $1-\frac{1}{t}$, it holds that: for all $s\in S$, $a\in A$, 
\begin{equation}\label{eq:empirical-quantity-good}
    \big|\hat \pi^{(t)}(s) \ru(a, \hat \mu_s^{(t)}) - \pi(s) \ru(a, \mu_s) \big| \le \sqrt{\frac{\log(2|S||A|t)}{2t}}. 
\end{equation}
We also note that $\hat \pi^{(t)}(s)$ is the average of $t$ i.i.d.~random variables in $[0, 1]$.  So, by Chernoff bound,
\begin{align*}
    \Pr\Big[ \big|\hat \pi^{(t)}(s) - \pi(s) \big| > \delta \pi(s) \Big] \le 2e^{-\frac{\delta^2 \pi(s) t}{3}}. 
\end{align*}
Using a union bound over $s\in S$ and letting $\delta = \sqrt{\frac{3 \log(2|S|t)}{\pi(s)t}}$, we have 
\begin{align*}
    \Pr\Big[\exists s\in S, \text{ s.t. } \big|\hat \pi^{(t)}(s) - \pi(s) \big| > \delta \pi(s) \Big] \le 2|S|e^{-\frac{\delta^2 \pi(s) t}{3}} = \frac{1}{t}. 
\end{align*}
In other words, with probability at least $1-\frac{1}{t}$, it holds that: for all $s\in S$, 
\begin{equation}\label{eq:empirical-pi-good}
    \frac{\big|\hat \pi^{(t)}(s) - \pi(s) \big|}{\pi(s)} \le \delta = \sqrt{\frac{3 \log(2|S|t)}{\pi(s)t}} \quad \iff \quad 1-\delta \le \frac{\hat \pi^{(t)}(s)}{\pi(s)} \le 1+\delta. 
\end{equation}
From \eqref{eq:empirical-pi-good} we obtain
\begin{align} \label{eq:empirical-pi-lower-bound}
    \frac{\pi(s)}{\hat \pi^{(t)}(s)} \ge \frac{1}{1 + \delta} \ge 1 - \delta = 1 - \sqrt{\frac{3 \log(2|S|t)}{\pi(s)t}}, 
\end{align}
and with the assumption that $\delta = \sqrt{\frac{3 \log(2|S|t)}{\pi(s)t}} < \frac{1}{2}$, 
\begin{align} \label{eq:empirical-pi-upper-bound}
    \frac{\pi(s)}{\hat \pi^{(t)}(s)} \le \frac{1}{1 - \delta} = 1 + \frac{\delta}{1-\delta} \le 1 + 2\delta = 1 + 2 \sqrt{\frac{3 \log(2|S|t)}{\pi(s)t}}. 
\end{align}
Divide \eqref{eq:empirical-quantity-good} by $\hat \pi_s^{(t)}$: 
\begin{align*}
    \big|\ru(a, \hat \mu_s^{(t)}) - \frac{\pi(s)}{\hat \pi^{(t)}(s)}\ru(a, \mu_s) \big| \le \frac{1}{\hat \pi^{(t)}(s)}\sqrt{\frac{\log(2|S||A|t)}{2t}}. 
\end{align*}
Using \eqref{eq:empirical-pi-lower-bound} and \eqref{eq:empirical-pi-upper-bound}, we get
\begin{align*}
    \ru(a, \hat \mu_s^{(t)}) & \ge \frac{\pi(s)}{\hat \pi^{(t)}(s)}\ru(a, \mu_s) - \frac{1}{\hat \pi^{(t)}(s)}\sqrt{\frac{\log(2|S||A|t)}{2t}} \\
    & \ge \Big(1 - \sqrt{\frac{3 \log(2|S|t)}{\pi(s)t}} \Big) \ru(a, \mu_s) - \frac{1}{\pi(s)} \Big(1 + 2 \sqrt{\frac{3 \log(2|S|t)}{\pi(s)t}} \Big) \sqrt{\frac{\log(2|S||A|t)}{2t}} \\
    & \ge \ru(a, \mu_s) - \sqrt{\frac{3 \log(2|S|t)}{\pi(s)t}} -  \frac{1}{\pi(s)} \cdot 2 \cdot \sqrt{\frac{\log(2|S||A|t)}{2t}} \\
    & = \ru(a, \mu_s) - \sqrt{\frac{3 \log(2|S|t)}{\pi(s)t}} -  \frac{2}{\pi(s)}\sqrt{\frac{\log(2|S||A|t)}{2t}}
\end{align*}
and
\begin{align*}
    \ru(a, \hat \mu_s^{(t)}) & \le \frac{\pi(s)}{\hat \pi^{(t)}(s)}\ru(a, \mu_s) + \frac{1}{\hat \pi^{(t)}(s)}\sqrt{\frac{\log(2|S||A|t)}{2t}} \\
    & \le \Big(1 + 2\sqrt{\frac{3 \log(2|S|t)}{\pi(s)t}} \Big) \ru(a, \mu_s) + \frac{1}{\pi(s)} \Big(1 + 2 \sqrt{\frac{3 \log(2|S|t)}{\pi(s)t}} \Big) \sqrt{\frac{\log(2|S||A|t)}{2t}} \\
    & \le \ru(a, \mu_s) + 2\sqrt{\frac{3 \log(2|S|t)}{\pi(s)t}} + \frac{1}{\pi(s)} \cdot 2 \cdot \sqrt{\frac{\log(2|S||A|t)}{2t}} \\
    & = \ru(a, \mu_s) + 2\sqrt{\frac{3 \log(2|S|t)}{\pi(s)t}} + \frac{2}{\pi(s)}\sqrt{\frac{\log(2|S||A|t)}{2t}}. 
\end{align*}
Thus, we obtain $|\ru(a, \hat \mu_s^{(t)}) - \ru(a, \mu_s)| \le 2\sqrt{\frac{3\log(2|S|t)}{\pi(s)t}} + \frac{2}{\pi(s)}\sqrt{\frac{\log(2|S||A|t)}{2t}}$.

\end{document}